\documentclass{jgaa-art}

\usepackage[utf8]{inputenc}
\usepackage{amsmath}
\usepackage{amssymb}
\usepackage{tabularx}
\usepackage[capitalize, noabbrev]{cleveref}
\usepackage{cite}
\usepackage{float}
\usepackage{gensymb}
\usepackage{subcaption}
\usepackage{xspace}
\usepackage{enumerate}
\usepackage{csquotes}
\usepackage{bbm}
\usepackage{mathtools}
\usepackage{wasysym}
\usepackage{lineno}
\usepackage[normalem]{ulem}
\usepackage[linesnumbered,ruled,vlined]{algorithm2e}

\usepackage{todonotes}

\newcommand{\winIgnore}[0]{}

\newcommand{\twocc}{2-crossing-critical graph\xspace}
\newcommand{\ltwocc}{large \twocc}
\newcommand{\Ltwocc}{\expandafter\MakeUppercase \ltwocc}

\newcommand{\Bracic}{Bračič\xspace}
\newcommand{\Dvorak}{Dvořák\xspace}
\newcommand{\Hlineny}{Hliněný\xspace}
\newcommand{\Siran}{Širáň\xspace}
\newcommand{\Zerak}{Žerak\xspace}

\newcommand{\tw}{\mathrm{tw}}
\newcommand{\crn}{\mathrm{cr}}
\newcommand{\scr}{\mathrm{cr}^\times}
\newcommand{\hg}{\raisebox{.4mm}{\rotatebox[origin=c]{90}{$\bowtie$}}\xspace}

\renewcommand{\emptyset}{\varnothing}

\newcommand{\sig}{\ensuremath{{\it sig}}}
\newcommand{\beautify}[1]{\ensuremath{\emph{#1}\xspace}}
\newcommand{\fL}{\beautify{L}}
\newcommand{\fdL}{\beautify{dL}}
\newcommand{\fd}{\beautify{d}}
\newcommand{\pA}{\beautify{A}}
\newcommand{\pV}{\beautify{V}}
\newcommand{\pB}{\beautify{B}}
\newcommand{\pH}{\beautify{H}}
\newcommand{\pI}{\beautify{I}}
\newcommand{\pD}{\beautify{D}}
\newcommand{\cC}{{\cal C}}
\newcommand{\cT}{{\cal T}}
\newcommand{\cS}{{\mathcal S}}

\newcommand{\bcirc}{\ocircle}

\newcommand{\change}[1]{#1}
\newcommand{\removed}[1]{}
\newcommand{\removedS}[1]{}

\newtheorem{definition}[theorem]{Definition}
\newtheorem{obs}[theorem]{Observation}
\newtheorem{corollary}[theorem]{Corollary}
\newtheorem{question}[theorem]{Question}
\numberwithin{theorem}{section}
\newtheorem{lem}[theorem]{Lemma}
\renewenvironment{lemma}{\begin{lem}}{\end{lem}}

\begin{document}

	\doi{}
	\Issue{0}{0}{0}{0}{0} %

	\title{Properties of Large \removed{$2$}\change{2}-Crossing-Critical Graphs}

	\HeadingAuthor{Bokal et al.} %
	\HeadingTitle{Properties of Large \removed{$2$}\change{2}-Crossing-Critical Graphs} %

	\author[UM]{Drago~Bokal}{drago.bokal@um.si}
	\author[UOS]{Markus~Chimani}{chimani@uos.de}
	\author[UOS]{Alexander~Nover}{anover@uos.de}
	\author[UOS]{Jöran~Schierbaum}{jschierbaum@uos.de}
	\author[UOS]{Tobias~Stolzmann}{tstolzmann@uos.de}
	\author[UOS]{Mirko~H.~Wagner}{mirwagner@uos.de}
	\author[UOS]{Tilo~Wiedera}{twiedera@uos.de}

	\affiliation[UM]{Dep.\ of Mathematics and Computer Science, University of Maribor, Slovenia}
	\affiliation[UOS]{Theoretical Computer Science, Osnabrück University, Germany}

	\submitted{}%
	\reviewed{}%
	\revised{}%
	\accepted{}%
	\final{}%
	\published{}%
	\type{Regular Paper}%
	\editor{}%

	\maketitle

	\begin{abstract}
		A $c$-crossing-critical graph is one that has crossing number at least~$c$ but each of its proper subgraphs has crossing number less than~$c$.
		Recently, a set of explicit construction rules was identified by Bokal, Oporowski, Richter, and Salazar to generate all \emph{large} $2$-crossing-critical graphs (i.e., all apart from a finite set of small sporadic graphs). They share the property of containing a generalized Wagner graph $V_{10}$ as a subdivision.

		In this paper, we study these graphs and establish their order, simple crossing number, edge cover number, clique number, maximum degree, chromatic number, chromatic index, and treewidth.
		We also show that the graphs are linear-time recognizable and that all our proofs lead to efficient algorithms for the above measures.
	\end{abstract}

	\Body

	\section{Introduction}
	The first characterization of \emph{planar} graphs is due to Kuratowski in 1930: A graph\change{\footnote{\change{Multiple edges and loops arise naturally in the context of graph embeddings and graph drawings. Hence, in such context, a graph can have multiple edges and loops, and the term simple graph is employed whereever we emphasize that these features are not present. We follow this convention throughout this paper.}}} is planar if and only if it neither contains a subgraph isomorphic to a subdivision of the $K_{3,3}$ nor the $K_5$~\cite{Kur30}.
	This result inspired several characterizations of graphs by forbidden subgraphs, which paved paths into significantly different areas of graph theory.
	Extremal graph theory is concerned with forbidding any subgraph isomorphic to a given graph~\cite{Bol86} and maximizing the number of edges under this constraint.
	Significant structural theory was developed when forbidden \emph{induced} subgraphs were considered instead, for instance several characterizations of Trotter~and~Moore~\cite{TM76} and the remarkable weak and strong perfect graph theorems~\cite{Ree87, Lov72}.
	Wagner coined graph minor theory as another means of characterizing planar graphs~\cite{Wag37}.
	It was later used to extend Kuratowski's theorem to higher surfaces: A seminal result by Robertson and Seymour states that all graphs embeddable into any prescribed surface are characterized by a finite set of forbidden minors~\cite{RS90}.
	While these minors are known for the projective plane~\cite{Arc81}, already on the torus, the number of forbidden minors reaches into tens of thousands and is as of now unknown~\cite{Gag05}.
	Still, Mohar devised an algorithm to embed graphs on surfaces in linear time~\cite{Moh99}, that was later improved by Kawarabayashi, Mohar, and Reed~\cite{Kaw08}.
	Characterizations of graph classes by subdivisions received somewhat less renowned attention.
	Early on the above path, Chartrand, Geller, and Hedetniemi pointed at some common generalizations of forbidding a small complete graph and a corresponding complete bipartite subgraph as a subdivision, resulting in trees, outerplanar, and planar graphs~\cite{CGH71}.
	More recently, \Dvorak achieved a characterization of several other graph classes using forbidden subdivisions~\cite{Dvo08}.

	Another direction to generalize Kuratowski's theorem is the notion of \emph{$c$-crossing critical} graphs, i.e., graphs that require at least $c\in\mathbb{N}$ crossings when drawn in the plane, but each of their subgraphs requires strictly less than~$c$ crossings.
	\change{Allowing crossings in order to increase the degree of freedom rather than adding handles to the surface exhibits a richer
	structure compared to forbidden minors for embeddability on surfaces.}
	Unlike the latter, it allows infinite families of topologically-minimal obstruction graphs,
	as first demonstrated by \Siran~\cite{Sir84}, who constructed an infinite family of
	$3$-connected $c$-crossing-critical graphs for each $c>2$.
	Kochol extended this result to simple, $3$-connected graphs~\cite{Koc87}, for each $c>1$, thus producing the first family of (simple) large $3$-connected $2$-crossing-critical graphs.
	Most importantly for our research is
	Bokal, Oporowski, Salazar, and Richter's \cite{Bok+16} characterization of the \emph{complete} list of minimal forbidden subdivisions for a graph to be realizable in the plane with only one crossing;
	that is, precisely the \emph{$2$-crossing-critical} graphs.
	Bokal, \Bracic, Dernar, and \Hlineny characterized average degrees for infinite families of $2$-crossing-critical graphs w.r.t.\ constraining the vertex-degrees that appear arbitrarily often \cite{Bok+19b}.
	For each restriction, the resulting average degrees form an interval.
	\Hlineny and Korbela showed that if \emph{all} degrees are prescribed, instead of just the frequent ones, the attainable average degrees are no longer intervals, but dense subsets of intervals \cite{HK19}.
	Based upon \cite{Bok+16},
	 Bokal, Vegi-Kalamar, and \Zerak defined a simple regular grammar describing \emph{large} 2-crossing-critical graphs---i.e., all 3-connected 2-crossing-critical graphs
	 except for a finite set of (small) sporadic graphs---and used it for counting Hamiltonian cycles in these graphs \cite{BKZ21}.
	We build upon this grammar to study the graph theoretic properties of \ltwocc{}s.

	We \removedS{may}also \change{briefly} discuss recognizing $c$-crossing-critical graphs. $1$-crossing critical graphs are precisely subdivisions of a $K_5$ or a $K_{3,3}$; they are thus trivial to recognize. For general $c\geq 2$, the problem is fixed-parameter tractable (FPT) w.r.t.~$c$: Grohe~\cite{Gro04} first showed that there is an algorithm to recognize graphs with $\crn(G)\leq c$ in time $\mathcal{O}(\mathit{poly}(|V(G)|)\cdot f(c))$ for some (at least doubly exponential but computable) function $f$. Kawarabayashi and Reed~\cite{KR07} improved this FPT-algorithm to an only linear dependency on $|V(G)|$. Despite the fact that these algorithms are infeasible in practice, they can theoretically be used as a building block to verify $\crn(G)\geq c$ and $\crn(G-e)<c$, for each $e\in E(G)$. Thus $c$-crossing-critical graphs can be recognized in FPT-time $\mathcal{O}(|V(G)| \cdot |E(G)| \cdot f(c))$, for some computable function $f$.
	We do not know of any further algorithmic results regarding the recognition problem.

	\paragraph{Our Contribution.}
	In \cref{sc:tiles}, we recall the formal definition of \ltwocc{}s, their construction, and the recently proposed grammar to chiefly describe them.
	In \cref{sc:elementary}, we proceed to determine some of their elementary properties, such as order, maximum degree, clique, and matching number.
	We also show that \ltwocc{s} are linear-time recognizable.

	In \cref{sc:simple-cr}, we establish that their \emph{simple} crossing number is indeed also~$2$.
	We propose sufficient sets of \emph{color propagations} (defined later) to find their chromatic number and index in \cref{sc:chromaticnum,sc:chromaticin}, respectively.
	Finally, in \cref{sc:treewidth}, we characterize the graphs' treewidth via the appearance of a single minor.
	Further, in each section, we propose natural linear time algorithms to compute the respective measures on any given \ltwocc{}.

	Although the graphs under consideration form a structurally rich, yet countable infinite family, our results underline their structural cohesiveness: all investigated measures reside in a small range, some are even constant over all such graphs.
	\Cref{tab:overview} summarizes all considered properties and our results.

	\begin{table}
		\caption{Overview on the properties of \ltwocc{}s studied in this paper. }
		\label{tab:overview}\centering
		\begin{tabular}{l@{\ }c@{\ }c@{}c}\hline
			property & characterization & values & see\\
			\hline\hline
			graph size & complete & -- & Observation~\ref{obs:order} \\
			maximum degree & complete & $4, 5, 6$ & Observation~\ref{obs:maxdeg} \\
			clique number & complete & $2, 3 $ & Corollary~\ref{cor:clique} \\
			edge cover number & complete & $\left\lceil \vert V(G) \vert/2 \right\rceil$ & Observation~\ref{obs:edgecover} \\
			simple crossing number & complete & $2$ & Theorem~\ref{thm:simple} \\
			chromatic number & partial & $2, 3, 4$ & Theorems~\ref{thm:CharacterizationBipartite} \& \ref{thm:4Colorability} \\
			chromatic index & complete$^*$ & $\Delta(G)$ & Theorem~\ref{thm:chromindex} \\
			treewidth & complete$^*$ & $4$, $5$  & Corollary~\ref{cor:bounded_treewidth}\\\hline
			\multicolumn{4}{m{.9\textwidth}}{\footnotesize $^*$ few (finitely many) graphs on only $3$ elementary tiles can attain smaller (resp. larger) values for treewidth (resp. chromatic index). See corresponding sections.}
		\end{tabular}
	\end{table}

	\section{Large 2-Crossing-Critical Graphs}
	\label{sc:tiles}

	For standard graph theory terminology, such as (induced) subgraphs and graph minors, we refer to \cite{Die06, Sch17}.
	\removedS{Recall that}\change{A} \emph{drawing} of a graph $G$ \emph{in the plane} consists of two injective maps:
	One assigning each vertex~$v \in V\change{(G)}$ to a point in~$\mathbb{R}^2$, the other each edge~$uv \in E\change{(G)}$ to a Jordan curve from~$u$ to~$v$ in $\mathbb{R}^2$ such that no curve has a vertex in its interior.
	In the context of crossing numbers, we typically restrict ourselves to \emph{good} drawings:
	Each pair of curves has at most one interior point in common (if it exists, it is the \emph{crossing} of this pair), adjacent curves have no \removedS{such}\change{common} crossing, and the intersection of any three non-adjacent curves is empty.

	\begin{definition}
		The \emph{crossing number~$\crn(G)$} of a graph~$G$ is the smallest number of crossings over all of its \removedS{plane}drawings \change{in the plane}.
		Further, $G$ is \emph{$c$-crossing-critical} \change{for some $c \in \mathbb{N}$,} if $\crn(G) \geq c$, but every proper subgraph~$H \subset G$ has~$cr(H)~<~c$.
	\end{definition}

	Following this definition, we feel that some intuitive explanation of the context is in place before we formalize the details in the rest of this section.
	Note that the above definition defines $2$-crossing-critical graphs, but does not say anything about how they actually look like.
	For $1$-crossing-critical graphs, this was resolved by Kuratowski's theorem, which exposed $K_5$ and $K_{3,3}$ as the only two $3$-connected
	$1$-crossing-critical graphs, and all other 1-crossing-critical graphs as their subdivisions. As mentioned in the introduction, already $2$-crossing-critical
	graphs---the next step beyond Kuratowski's Theorem---exhibit a significantly richer structure, and allow for an infinite family of $3$-connected $2$-crossing-critical graphs \change{\cite{Sir84,Koc87}}. However, despite being infinite, this family has a tightly defined structure. The purpose of this section is to
	describe this structure, i.e., to use the \change{characterization} results of \cite{Bok+16} to explain how (almost all) $3$-connected $2$-crossing-critical graphs actually look like.
	All $3$-connected $2$-crossing-critical graphs \change{with sufficiently many vertices} exhibit this structure;
	only finitely many do not (\change{the} Petersen graph being the most
	prominent example). Thus we call the graphs having this structure \emph{large $2$-crossing-critical graphs}.
	Let us formally define the construction rules that generate the set of \twocc{}s and give a brief overview of their history.
	The concept of \emph{tiles} (to be defined in this section) was introduced by Pinontoan and Richter
	to answer a question of Salazar about average degrees of large families of $c$-crossing-critical graphs~\cite{PR04, Sal03}.
	Over a series of papers, it turned out to be a tool that gives surprisingly precise lower bounds on crossing numbers of
	several ``tiled'' graphs, see \cite{Bok+19a}.
	\Dvorak, \Hlineny, and Mohar showed that tiles form an essential ingredient of large
	$c$-crossing-critical graphs for every $c\ge 2$~\cite{DHM18}.
	In general, further structures (so-called \emph{belts} and \emph{wedges}) may also appear arbitrarily often, together with a bounded small graph that connects them \cite{Bok+19a}.
	For $c=2$, however, Bokal, Oporowski, Richter, and Salazar proved that tiles are sufficient to describe almost all (i.e., all but finitely many)
	\twocc{}s \cite{Bok+16}.
	In fact, belts appear if and only if $c\ge 3$ and wedges if and only if $c\ge 13$ \cite{Hli01, Bok+19a}.

	Intuitively, \emph{tiles} are prespecified small graphs with vertex subsets at which we can glue (pairs of) tiles together. A \emph{tiled graph} is a graph arising from cyclically glueing tiles together. Formally, we adopt the following notation from \cite{Sal03}\change{, which is illustrated in Figure 3}:

	\begin{definition}
		\label{def:kTile}
		A \emph{tile} is a triple $T=(G,x,y)$, consisting of a graph $G$ and two non-empty
		sequences $x = \langle x_1, x_2, \ldots, x_k \rangle$ and $y = \langle y_1, y_2, \ldots, y_l \rangle$ of distinct vertices of $G$,
		with no vertex appearing in both $x$ and $y$.
		The sequence~$x$ (sequence~$y$) is $T$'s \emph{left wall} (\emph{right wall}, resp.).
		If $|x| = |y| = k$, $T$ is a \emph{$k$-tile}.
	\end{definition}

	\begin{definition}
		\label{def:kTiledGraph}
		\emph{Tiled graphs} are joins of cyclic sequences of tiles. We formalize this as follows:
		\begin{enumerate}
			\item
			A tile~$T=(G,x,y)$ is \emph{compatible} with a tile~$T'=(G',x',y')$  if $|y| = |x'|$.
			Their \emph{join}~$T\otimes T'\coloneqq(G^*,x,y')$ is a new
			tile, where $G^*$ is obtained from the disjoint union of $G$ and $G'$
			by identifying $y_i$ with $x'_i$ for each $i=1,\ldots,|y|$.
			\item
			A sequence $\cT = \langle T_0, T_1,\ldots, T_m \rangle$ of tiles is \emph{compatible} if
			$T_{i-1}$ is compatible with $T_{i}$ for each $i = 1,2,\ldots,m$.
			The \emph{join}~$\otimes \cT$ of a compatible sequence~$\cT$ is
			$T_0 \otimes T_1 \otimes \cdots \otimes T_m$.
			\item
			For a $k$-tile $T=(G,x,y)$,
			the \emph{cyclization} of $T$ is the graph~$\bcirc T$
			obtained from $G$ by identifying  $x_i$ with $y_i$ for each $i=1,\ldots,k$. (Observe that in general, $T$ may itself have arisen from a join of a compatible sequence.)
		\end{enumerate}
	\end{definition}

	With these tools, we are now ready to recall the constructive characterization of \twocc{}s by tiles \cite{Bok+16}\change{. T}hereby, we focus only on the graphs that
	belong to the theoretically relevant infinite family of these graphs\change{. W}e disregard some finite set of special cases as well as graphs that are not $3$-connected\change{, as they add no relevant structural information}.
	The latter ones can be trivially obtained from the $3$-connected ones, and 3-connectivity is
	a typical restriction when studying graphs from a topological perspective, such as crossing numbers.
	Put chiefly, we provide the characterization of all (except for finitely many) $3$-connected \twocc{}s, called \ltwocc{}s.
	In the course of this, we will also describe these graphs' alphabetic description \cite{BKZ21}, which associates unique \removedS{u}\change{a}nd coherent names to each such graph.

	Again, before we formally define the set $\mathcal{C}$ of \removedS{all}\ltwocc{}s, we may give an intuitive definition. There are 42 planar 2-tiles $\mathcal{S}$ to choose from (to be described later).
	Each graph in $\mathcal{C}$ is a cyclization of a sequence of an odd number of these tiles. But thereby, every second tile will be used flipped top-to-bottom (which is not so important right now), and we reverse the order
	of the final right wall vertices prior to the cyclization. Without this final twist, the resulting graph would resemble a cyclic planar strip of tiles; due to the twist, the resulting graph becomes non-planar but can be embedded on the Möbius strip.

	\begin{figure}[tb]
		\centering
		\begin{subfigure}[b]{0.55\textwidth}
			\centering
			\winIgnore{\includegraphics[scale=0.5]{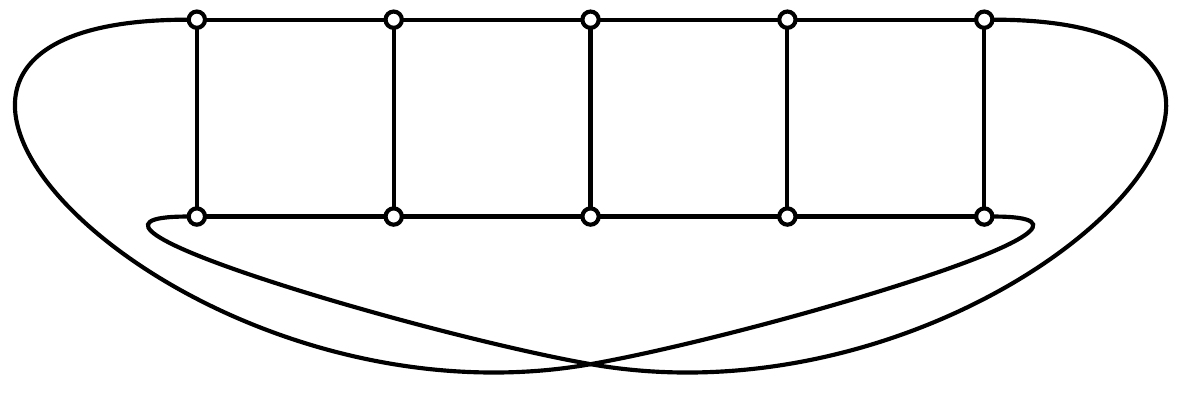}}
			\caption{$V_{10}$ drawn as Möbius ladder.
			}
			\label{fig:ML}
		\end{subfigure}
		\begin{subfigure}[b]{0.44\textwidth}
			\centering
			\winIgnore{\includegraphics[scale=0.5]{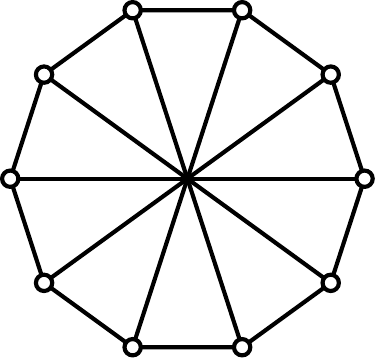}}
			\caption{$V_{10}$ drawn according to definition.}
			\label{fig:WG}
		\end{subfigure}
		\caption{The generalized Wagner graph $V_{10}$.}
		\label{fig:MoebiusWagner}
	\end{figure}

	Figure~\ref{fig:MoebiusWagner} shows the general graph structure exhibited by this process: assume each tile is drawn within a square region, then \cref{fig:ML} represents the resulting Möbius strip, where one of the squares is twisted. The graph that is depicted is in fact a generalized Möbius ladder, also known as generalized Wagner graph, and it is instrumental in understanding 2-crossing-critical graphs. Formally, it is defined as the graph $V_{2n}$, $3 \le n \in\mathbb{N}$, obtained from the cycle $C_{2n}$ in which each pair of antipodal vertices is connected via an additional edge (a \emph{spoke} of $V_{2n}$), see \cref{fig:WG}. The smallest Wagner graph $V_6$ is isomorphic to $K_{3,3}$.

	Based on this structure, assuming each tile $T_i$ in the sequence has some unique string $s_i$ as its name, it is straight-forward to use the concatenation $s_1s_2\ldots$ to describe the resulting graph. We call these strings \emph{signatures}. The join of our tiles can also be understood such that we cyclically join tiles (without vertical flipping) by \emph{always} reversing the order of the right wall vertices. While this understanding is not very helpful in terms of drawings with low crossing number, it shows that the graph-defining tile sequence is intrinsically cylic; consequently each graph's signature can be cyclically rotated as well, and for a graph with $k$ tiles we obtain $k$ potentially different signatures.

	It remains to discuss the fundamental 42 planar 2-tiles themselves, as they are highly structured. Each tile can be understood to be composed of a \emph{frame} and a \emph{picture} within that frame. Formally, these are \removed{multi}graphs, enriched with vertex markings. There are two different frames (\cref{fig:tileFrames}), and 21 different pictures (\cref{fig:tilePictures}). We will hence compose the signature of a tile as the concatenation of signatures of its picture and its frame. The names of the pictures arise from the graph structures along the top and bottom border of the tile (\emph{top path} and \emph{bottom path}, respectively) and their rough similarity to letters; see \cref{fig:tilePathsTop}.

	The example graph on five tiles in \cref{fig:2cc_example} completes the informal definition of the construction of \ltwocc{}s. We will revisit this example graph in later sections to showcase the investigated properties. We are now ready to formally define our graph class.

	\begin{figure}[p]
		\centering
		\begin{subfigure}[c]{\textwidth}
			\centering
			\winIgnore{\includegraphics[width=0.45\textwidth]{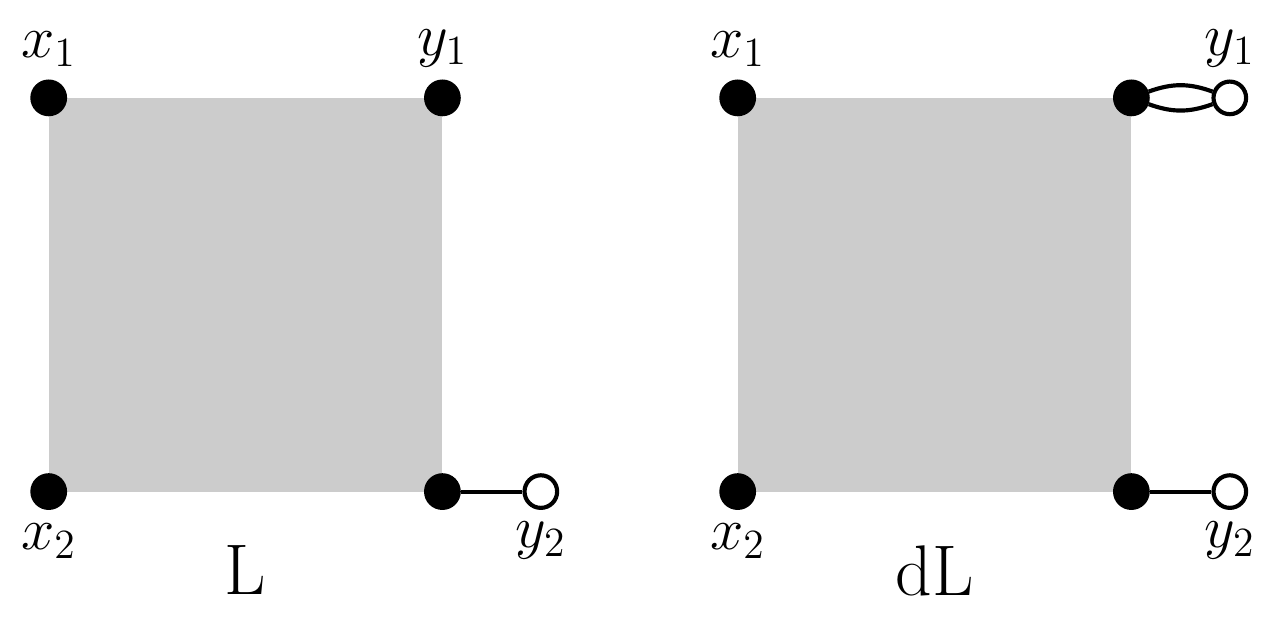}}
			\caption{The two frames and their names. Left wall vertices of the tiles obtained from these frames are $x_1,x_2$ and right wall vertices are $y_1,y_2$.}
			\label{fig:tileFrames}
		\end{subfigure}

		\bigskip

		\begin{subfigure}[c]{\textwidth}
			\centering
			\winIgnore{\includegraphics[width=0.8\textwidth]{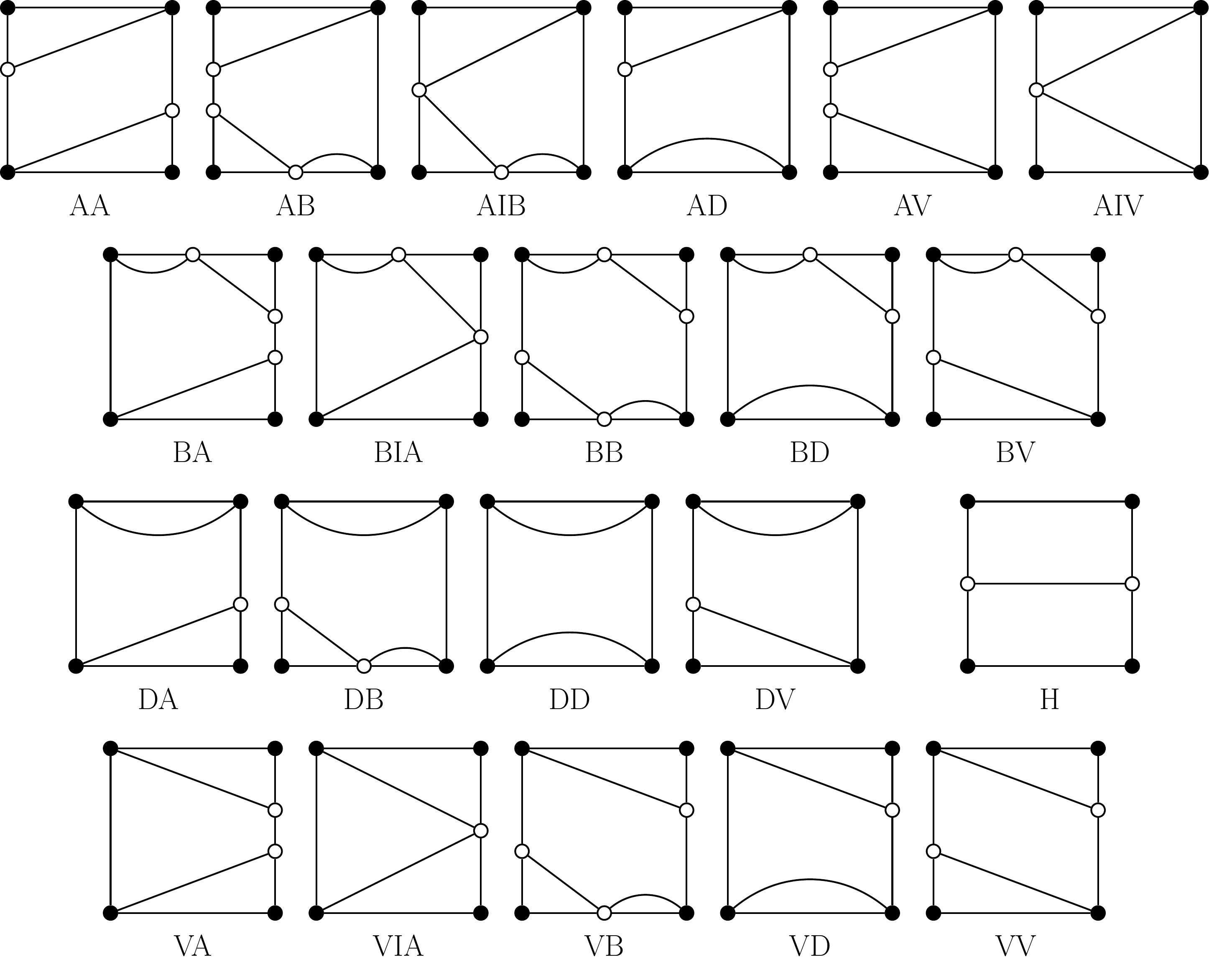}}
			\caption{The $21$ pictures and their names. The black vertices of the pictures are identified with black vertices of the frames (without additional rotation) to yield tiles in the set~$\mathcal S$.}
			\label{fig:tilePictures}
		\end{subfigure}

		\bigskip

		\bigskip

		\begin{subfigure}[c]{\textwidth}
			\centering
			\winIgnore{\includegraphics[width=0.8\textwidth]{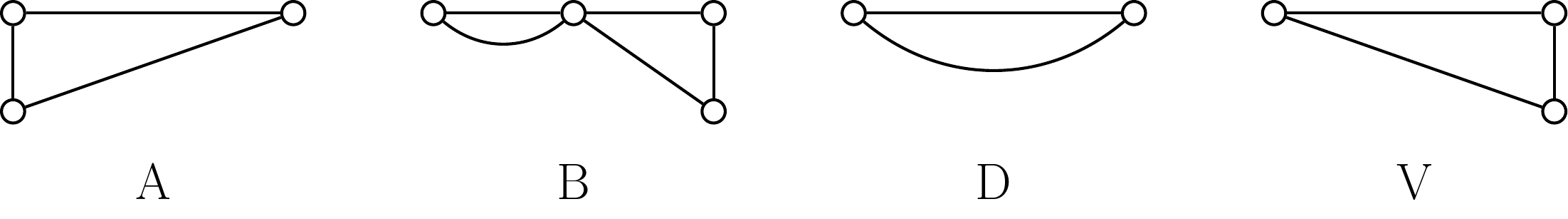}}
			\caption{Explanation for the names of the pictures (except for~\pH\ where the basis for the name is evident): We show the \emph{top path} of the pictures, together with their names.
				\emph{Bottom paths} are referred to equivalently but rotated by $180\degree$. The amalgamation of these names yield the signatures of the pictures, using the additional letter \pI, if a vertex of the top path becomes identified with one of the bottom part.
				\label{fig:tilePathsTop}}
		\end{subfigure}

		\bigskip

		\caption{Composition of tiles by pasting pictures into frames. The black vertices are identified when inserting a picture into a frame at the gray square. The tile's wall vertices are labeled.
		}
		\label{fig:tileFrameContent}
	\end{figure}

	\begin{figure}[t]
		\begin{subfigure}{\textwidth}
			\centering
			\winIgnore{\includegraphics[width=\textwidth]{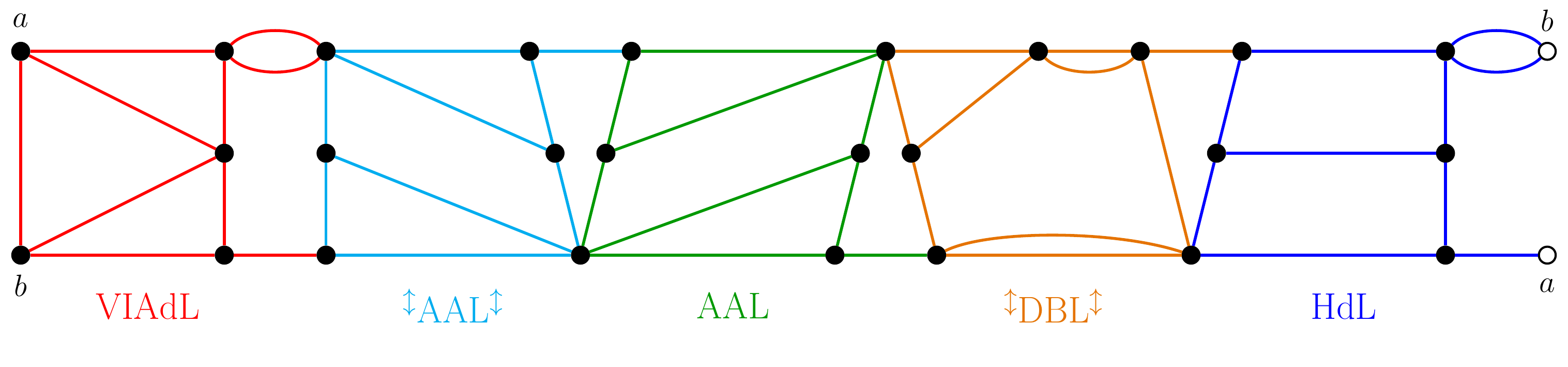}}
		\end{subfigure}
		\begin{subfigure}{\textwidth}
			\centering
			\winIgnore{\includegraphics[scale=.7]{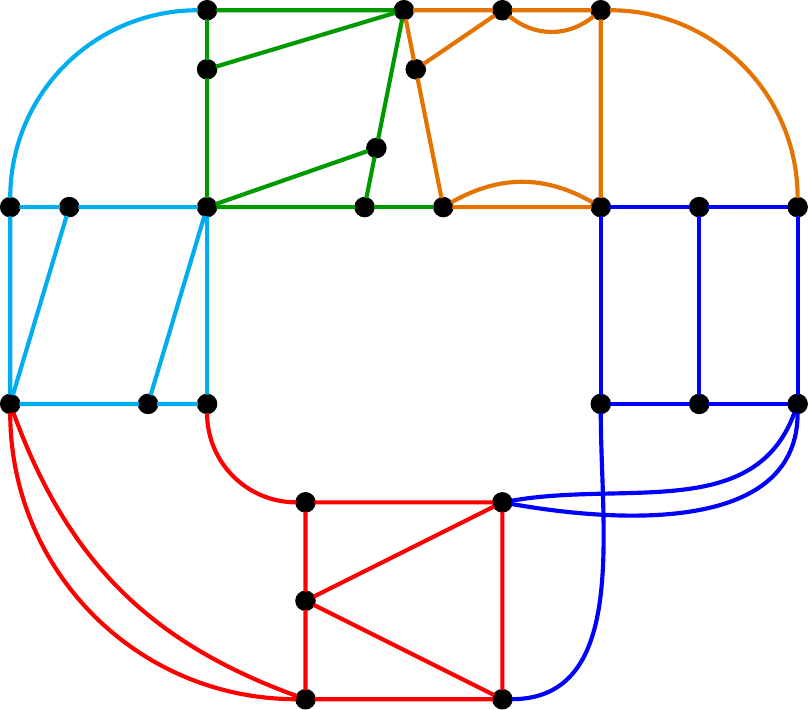}}
		\end{subfigure}

		\caption{An example of a \ltwocc; its signature is \mbox{\emph{VIAdL\,AAL\,AAL\,DBL\,HdL}}; colors represent the individual elementary tiles. On the top, the graph is drawn on the projective plane, where the labeled vertices are identified according to their names. On the bottom, the same graph is drawn equivalently but in the plane, resulting in two explicit crossings when twisting the dark blue HdL-tile.}
		\label{fig:2cc_example}
	\end{figure}

	\begin{definition}[based on \cite{Bok+16}]
		\label{def:constructionLarge2ccGraphs}
		Large $2$-crossing-critical graphs are defined as the set $\cC$ in the following way:
		\begin{enumerate}
			\item For a sequence $x$, let $\overline x$ denote the reversed sequence.
			The \emph{right-inverted} (\emph{left-inverted}) tile of a tile $T = (G, x, y)$ is the tile $T^{\updownarrow} \coloneqq (G, x, \overline{y})$ (and $^{\updownarrow}T \coloneqq (G, \overline{x}, y)$, respectively).
			\item \label{it:t3}
			Let $\cS$ be the set of tiles obtained as combinations of one of the two frames and one of the 21 pictures, shown in \cref{fig:tileFrameContent}, in such a way that a picture is inserted into a frame by identifying the gray area with it; the picture may not be rotated.
			While disregarding whether any wall order is reversed, we may call the tiles of $\cS$ \emph{elementary tiles}.
			\item \label{it:t4}
			Let $\mathcal C$ denote the set of all graphs of the form
			$\bcirc (T_0^{\updownarrow} \otimes T_1^{\updownarrow} \otimes \ldots \otimes T_{2m}^{\updownarrow})$
			with $m \geq 1$ and each $T_i \in \cS$.
			\item The \emph{signature~$\sig(T)$ of a tile~$T\in\cS$} is the concatenation of the names of its picture and its frame.
			A \emph{signature of a graph~$G$} is based on its tile construction: $\sig(G)\coloneqq\sig(T_0)\sig(T_1)\ldots\sig(T_{2m-1})\sig(T_{2m})$.
		\end{enumerate}
	\end{definition}

	Observe that, by cyclic symmetry, the signature of a graph in $\mathcal C$ is not unique.
	Given two tiles $T_a,T_b$, also observe that $T_a^{\updownarrow}\otimes T_b$ is isomorphic to $T_a\otimes {}^{\updownarrow}T_b$. Thus we can rewrite
	$\bcirc (T_0^{\updownarrow} \otimes T_1^{\updownarrow} \otimes \ldots \otimes T_{2m}^{\updownarrow})=
	 \bcirc((T_0 \otimes {}^{\updownarrow}T_1^{\updownarrow} \otimes T_2 \otimes \ldots \otimes {}^{\updownarrow}T_{2m-1}^{\updownarrow} \otimes T_{2m} )^{\updownarrow})$.
	While the former is formally more appealing and highlights the intrinsic symmetry, the latter implicitly tells us how to draw the graph with only 2 crossings: vertically flip every second tile to avoid all crossings until the last tile, where we require a simple twist.

Note that Definition \ref{def:constructionLarge2ccGraphs}
		does not imply that these graphs are actually $2$-crossing-critical, but the following theorem does:
	\begin{theorem}[{{Characterization by tiles \cite[Theorems 2.18 \& 2.19]{Bok+16}}}]
		\label{thm:3con2cc}
		Each element of~$\mathcal C$ is $3$-connected and $2$-crossing-critical.
		Furthermore, all but finitely many $3$-connected $2$-crossing-critical
		graphs are contained in $\mathcal C$, and the set $\mathcal C$ contains all the $2$-crossing-critical graphs that contain a $V_{10}$ subdivision.
	\end{theorem}

	\removed{In fact, all (but finitely many) 2-connected 2-crossing-critical graphs are
	obtained from C by subdividing edges and/or replacing double-edges by a path
	of double-edges [10]. Whenever it serves our investigation, we may assume any
	lower bound on the number of elementary tiles in large 2-crossing-critical graphs,
	disregarding a slightly larger but still finite set of graphs.\\}\change{Note that there may be small graphs of $\mathcal C$ that are 3-connected 2-crossing-critical,
	but do not have a $V_{10}$ subdivision. Although this is a rather technical challenge, understanding 
	it may simplify some approaches and the definition of $\mathcal C$, hence we pose it as an open problem: 
	\begin{question}
	List graphs of $\mathcal C$ with smallest number of vertices and edges. 
	List graphs in $\mathcal C$ that do not contain a $V_{10}$ subdivision or show that there are none.
	\end{question}
	}

	We denote the number of occur\change{r}ences of a given symbol~$X \in \change{\{\pA,\pV, \pD, \pB, \pH, \pI\}}$ \removedS{$\Sigma$}in the signature of a \twocc $G$ by $\#X(G)$.
	We may omit the parameter~$G$ if it is clear from the context.
	It is trivial to test in linear time whether a \removedS{word}\change{supposed signature indeed} describes a large $2$-crossing-critical graph.%

\section{Elementary Properties}
\label{sc:elementary}
Given the characterization of large $2$-crossing-critical graphs, we start our study by analyzing their elementary properties.
We will later use these results to facilitate the study of more involved measures.

\begin{obs}
    \label{obs:order}
	The number of vertices and edges of a \ltwocc $G$ is \change{obtained using the following matrix-vector multiplication}:
\end{obs}

\begin{equation*}
    \begin{bmatrix}
		|V\change{(G)}| \\
		|E\change{(G)}|
    \end{bmatrix}
    \coloneqq
    \begin{bmatrix}
        3 & 5\\
        1 & 2\\
        1 & 2\\
        1 & 2\\
        0 & 1\\
        2 & 3\\
        2 & 4\\
        -1 & -1
    \end{bmatrix}^\intercal\!
    \cdot\,
    \begin{bmatrix}
        \text{\#\fL} \\
        \text{\#\fd} \\
        \text{\#\pA} \\
        \text{\#\pV} \\
        \text{\#\pD} \\
        \text{\#\pH} \\
        \text{\#\pB} \\
        \text{\#\pI}
    \end{bmatrix}
\end{equation*}

\begin{proof}
Considering each elementary tile, we count its number of vertices and edges.
Since joining two tiles reduces the number of vertices by $2$, we reduce the number of vertices for each tile by $2$ (recall that this join is cyclic). It is straightforward to verify that each tile's signature generates the correct number of vertices and edges.
\end{proof}

Our example graph in \cref{fig:2cc_example} with $\sig(G)=\mbox{\textit{VIAdLAALAALDBLHdL}}$ yields the graph-dependent vector $[5,2,5,1,1,1,1,1]^\intercal$. Thus $[|V\change{(G)}|,|E\change{(G)}|]^\intercal = [3\cdot5 + 1\cdot2 + 1\cdot5 + 1\cdot1 + 0\cdot1 + 2\cdot1 + 2\cdot1 -1\cdot1, 5\cdot5 + 2\cdot2 + 2\cdot5 + 2\cdot1 + 1\cdot1 + 3\cdot1 + 4\cdot1 -1\cdot1]^\intercal=[26,48]^\intercal$.

\begin{obs}\label{obs:maxdeg}%
 The maximum degree $\Delta$ of a \ltwocc $G$ satisfies $4 \leq \Delta \leq 6$. In particular:
 \begin{itemize}
  \item $\Delta(G)=6$ if and only if there are two consecutive elementary tiles $T_1,T_2$, such that $T_1$'s frame is \fL, its top path is $\pA$ or $\pD$, as is the bottom path of~$T_2$
  (these paths are not necessarily equal).
  \item $\Delta(G)=5$ if and only if $\Delta(G)\neq6$ and $\#\pA +\#\pD > 0$.
  \item $\Delta(G)=4$ if and only if $\#\pA+\#\pD=0$.
 \end{itemize}
\end{obs}
\begin{proof}
    All elementary tiles with a $\fdL$-frame have a vertex of degree at least $4$ where the frame's double edge connects. All elementary tiles with an $\fL$-frame have a vertex of at least degree $2$ in the \enquote{top right}, which gets identified with a vertex of the next tile with degree at least $2$. Therefore, our graphs always contain vertices of degree at least $4$.

    Any path of $\{\pA,\pB,\pD,\pV,\pH\}$ increases the degree of vertices it connects to by at most~$1$. If a tile has top path $\pA$ or $\pD$, its top right vertex has degree at least~$5$. The same applies to a tile's bottom left vertex, if it has bottom path $\pA$ or $\pD$. Only by having a tile with an $\fL$ frame and an upper path $\pA$ or $\pD$ followed by a tile with bottom path $\pA$ or $\pD$, the identified vertex's degree becomes $6$.
\end{proof}

    A \emph{clique} in a graph $G$ is a subgraph of $G$ that is complete. The \emph{clique number} of $G$ is the order of the maximum clique.
\begin{obs} \label{lemma:picture_triangle}\label{lemma:graph_triangle}
    An elementary tile contains a triangle if and only if its signature contains \pA, \pV, or \pB\ (cf. \cref{fig:tileFrameContent}).
	Moreover, each triangle in a \ltwocc $G$ \removedS{that is not contained in an elementary tile}\change{whose signature does not contain \pA, \pV, or \pB} corresponds to a \pD\pD\fL\pD\pD-subsequence of $\sig(G)$.
\end{obs}

\noindent In fact, this observation is sufficient to fully determine the clique number \change{of a \ltwocc}.
\begin{obs}
 A \ltwocc contains no $K_4$.
\end{obs}

\begin{corollary}\label{cor:clique}
A \ltwocc has clique number $2$ if and only if all elementary tiles of $G$ are one of $\pD\pD\fL$, $\pD\pD\fdL$, $\pH\fL$ and $\pH\fdL$, and no subsequence $\pD\pD\fL\pD\pD$ exists in $\sig(G)$. Otherwise, its clique number is $3$.
\end{corollary}

\begin{corollary}
	Given the signature of a \ltwocc, its clique number can be determined in linear time.
\end{corollary}

	A \emph{matching} in a graph \change{$G$} is a subset of pairwise non-adjacent edges. %
It is \emph{perfect (near-perfect)} if its cardinality is $|V\change{(G)}| / 2$ ($(|V\change{(G)}|-1)/2$, resp.).
From the fact that each \ltwocc contains a Hamiltonian cycle which can be computed in linear time \cite{BKZ21}, we obtain:
\begin{obs}\label{obs:edgecover}
	Any \ltwocc~$G$\removed{~= (V,E)} has a perfect matching if $\vert V(G) \vert$ is even, and a near-perfect matching otherwise. In both cases, the matching can be computed in linear time by choosing every second edge of a Hamiltonian cycle.
\end{obs}
\begin{definition}
	The \emph{edge covering number} of a graph $G$ is the minimal number of edges~$F$ in $G$ such that each vertex~$v \in V(G)$ is incident to an edge in $F$.
\end{definition}
Since a perfect matching yields a minimum edge cover, and a near-perfect matching requires only a single additional edge to become an edge cover, we have:
\begin{obs}
	The edge covering number of any large 2-crossing-critical graph is $\left\lceil \vert V\change{(G)} \vert/2 \right\rceil$.
\end{obs}

Most importantly, \ltwocc are linear time recognizable.
The general idea of Algorithm \ref{alg_l2cc_recognition} is to restrict ourselves to a linear number of constantly sized graphs; in each of them, finding elementary tiles only requires constant time.
In particular, this algorithms allows us to, in linear time, deduce the signature of a given graph if it belongs to the class; as such it will be the starting point for all subsequent algorithms to compute graph properties, as they can thus assume to be given the signature as input.

\SetKwFor{Loop}{loop}{}{}
\begin{algorithm}[tb]
	\KwIn{graph $G$}
	\KwOut{$\sig(G)$ if $G$ is a \ltwocc; $\emptyset$ otherwise}
	\SetKwInput{KwDefine}{Define}
	\KwDefine{$G\!\!\downarrow^v_8$ is the subgraph of $G$ marked by a breadth-first search of bounded depth 8 starting at vertex $v$.}
	\SetKw{KwBreak}{break} %
	\DontPrintSemicolon
	\SetAlgoNoEnd %
	\BlankLine
	\If{$\Delta(G) > 6$}{ \KwRet{$\emptyset$\hfill// $G$ is not a \ltwocc} }
	choose any vertex $v\in V(G)$\;
	let $\mathcal{X} := \{X \subseteq G\!\!\downarrow^{v}_8\ : X \textit{~is \change{tile-}isomorphic to a tile in $\mathcal{S}$}\}$ \;
	\ForEach{$X \in \mathcal{X}$}{
		$\sig(G) := \sig(X)$ \;
		$G^* :=$ $G$ without the non-wall vertices of $X$\;
		$w_T, w_B :=$ the top and bottom right wall vertices of $X$\;
			\Loop{}{
				Search for a subgraph $Y\subseteq G^*\!\!\downarrow^{x_1}_8$ that is \change{tile-}isomorphic to a tile in~$\mathcal{S}$ using $w_B$ as top left and $w_T$ as bottom left wall vertex (note the reversed order); prefer one with a $\fdL$-frame over one with an $\fL$-frame\;
			\lIf(\hfill\emph{// continue with next $X$}){$\nexists~Y$}{\KwBreak}
			append $\sig(Y)$ to $\sig(G)$\;
			$w_T, w_B :=$ the top and bottom right wall vertices of $Y$\;
			remove all vertices of $V(Y)\setminus\{w_T,w_B\}$ from $G^*$\;
			\If{$V(G^*)\setminus\{w_T,w_B\} = \emptyset$}{
				\If{$w_B,w_T$ are the top and bottom left wall vertices of $X$ (note the reversed order) and $\#L(\sig(G))$ is odd}
				{\KwRet{$\sig(G)$} \hfill\emph{// $G$ is a \ltwocc}}
				\lElse{\KwBreak \hfill\emph{// continue with next $X$}}
			}
		}
	}
	\KwRet{$\emptyset$ \hfill// $G$ is not a \ltwocc}
	\caption{\Ltwocc recognition algorithm, deducing the signature in the positive case.}
	\label{alg_l2cc_recognition}
\end{algorithm}

\begin{theorem}
	\removedS{Given a graph $G$, Algorithm~\ref{alg_l2cc_recognition} is a linear-time algorithm to check whether}\change{Algorithm~\ref{alg_l2cc_recognition} tests in linear-time whether a given graph} $G$ is a \ltwocc and, in the positive case, deduces a signature of~$G$.
\end{theorem}
\begin{proof}
	We can reject graphs with maximum degree $\Delta(G)>6$ in linear time (line~1).
	\change{We say a subgraph $H$ of $G$ is \emph{tile-isomorphic} to a tile $T$, if $H$ is isomorphic to $T$, the non-wall vertices of $H$ have no neighbors other than those described by $T$, and wall vertices of $H$ are only adjacent if they are adjacent in~$T$.}
	We compute a subgraph $G\!\!\downarrow^v_8$ via a breadth-first search of bounded depth 8 starting at some arbitrary vertex $v$.
	This subgraph has constant size and can be found in constant time since $\Delta(G) \leq 6$.
	Furthermore, as the number $|\mathcal{S}|$ of possible tiles is constant, we can find the (constantly sized) set $\mathcal{X}$ of all subgraphs of $G\!\!\downarrow^v_8$ that are \change{tile-}isomorphic to a tile of $\mathcal{S}$ in constant time as well (line~4).
	The depth 8 is chosen so that, if $G$ is a \ltwocc, it is guaranteed that  $G\!\!\downarrow^v_8$ contains some subgraph $H$ is \change{tile-}isomorphic to a tile of $\mathcal{S}$; thus $H\in\mathcal{X}$. The set $\mathcal{X}$ thus serves as a candidate list for $T$. We run the subsequent test (lines 6--18) for each $X\in\mathcal{X}$ (for-loop starting at line 5):

	We remove $X$, retaining its wall vertices, and look for the neighboring tile~$Y$ to the right. Again, this search only requires constant effort (line 10).
	In the positive case,
	after removing all of $Y$ except its right wall vertices, we can iterate this process to identify all subsequent neighboring tiles, until we reattach -- after an overall odd number of tiles -- to the left wall of the initial tile $X$. By definition, we have to assure that the subsequent tiles use the common wall vertices in reverse order.
	If this process fails at any point, we reject the starting tile $X$ and proceed with the next iteration of the for-loop, i.e., the next candidate from~$\mathcal{X}$. If no iteration of the for-loop succeeds, we reject $G$.

	In each iteration of the inner loop (lines 9--18) we either terminate the current for-loop iteration or remove a constant number of edges.
	Thus, the inner loop runs at most a linear number of times, each of its iterations requiring only constant time.
	This establishes the overall linear running time.

	It is easy to see that if the algorithm returns a non-empty signature $s$, the \ltwocc constructed from $s$ as per Definition~\ref{def:constructionLarge2ccGraphs} is isomorphic to $G$.
	On the other hand, suppose $G$ is a \ltwocc and let $s$ be a signature of $G$, such that $v$ is in the first elementary tile~$T\in\mathcal{S}$ of~$s$.
	From the definition of $\mathcal{X}$, it follows that $T \in \mathcal{X}$.
	We only need to focus on the for-loop iteration in which $X = T$.
	If a graph has an elementary tile with a $\fdL$-frame as a subgraph, it also has an elementary tile with the same picture but an $\fL$-frame as a subgraph, but the converse is not true.
	Also given two elementary tiles with distinct pictures, at most one of them can be a subgraph that can be a right neighbor of the previously identified tile.
	Thus, in line 10, we obtain a unique potential candidate by prefering the new neighboring tile $Y$ to have a $\fdL$-frame if possible.
	Based on the structure that tiles are cleanly separated by wall vertices (see Def.~\ref{def:kTiledGraph} and \ref{def:constructionLarge2ccGraphs}), we consequently have that our algorithm will indeed find signature $s$.
\end{proof}

\section{Simple Crossing Number}
\label{sc:simple-cr}

In this section, we prove that the simple crossing number of each large $2$-crossing-critical graph equals its crossing number. To this end, we provide some definitions and briefly discuss their history.
The study of \emph{$1$-planar} graphs was initiated more than half a century ago by Ringel in the context of graph coloring~\cite{Rin65}.
Buchheim et al.\ introduced the \emph{simple crossing number}, while engineering the first general exact algorithms for computing crossing numbers~\cite{Buc+05}.

\begin{definition}%
	A \emph{$1$-planar drawing} of a graph~$G$ is a drawing of~$G$ in the plane such that each of its edges crosses at most one other edge.
	A graph that admits a $1$-planar drawing is called \emph{$1$-planar}.
	The \emph{simple crossing number} ~$\scr(G)$ of $G$ is the minimal number of crossings over all $1$-planar drawings of $G$; we define $\scr(G) = \infty$ if no such drawing exists.
\end{definition}

Albeit this crossing number variant is also known as \emph{1-planar} crossing number, we prefer the term simple crossing number. This avoids confusion with the \emph{$k$-planar} crossing number, $k\in\mathbb{N}$, as defined by Owens (there, the graph is partitioned into $k$ edge-sets and only the crossings in each set are counted)~\cite{Owe71, Sch13}.

By definition, $\crn(G) \leq \scr(G)$. We remark that in general, $\crn(G) \neq \scr(G)$ and there are graphs~$G$ with $\crn(G) = 2$
but $\scr(G) > 2$ even on as few as $16$~vertices~\cite{Buc+05}.
By definition $\scr(G) \in O(n^2)$, but for example $\crn(K_n)\in\Theta(n^4)$;
in fact, already $K_7$ is not $1$-planar (see, e.g., \cite{Sch13,Sch17}).

\begin{figure}
	\centering
	\begin{subfigure}[t]{.48\textwidth}\centering
		\winIgnore{\includegraphics[scale=.65]{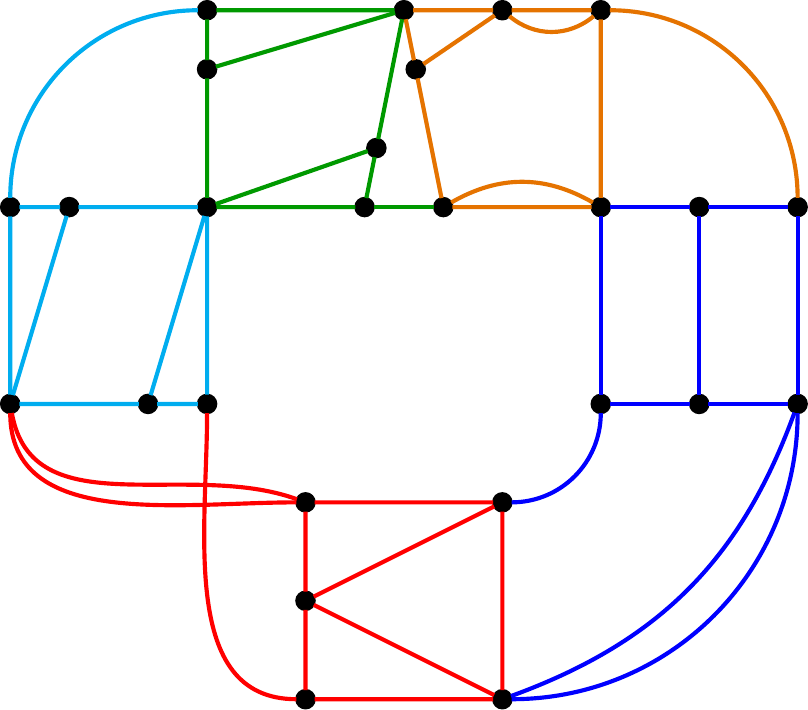}}
		\subcaption{$G$ still drawn as a Möbius strip but with the double crossing in the
			\emph{VIAdL}-tile (red) instead of the \emph{HdL}-tile (dark blue), to make it more similar to the figure to the right.}
	\end{subfigure}
	\hfill
	\begin{subfigure}[t]{.48\textwidth}\centering
		\winIgnore{\includegraphics[scale=.65]{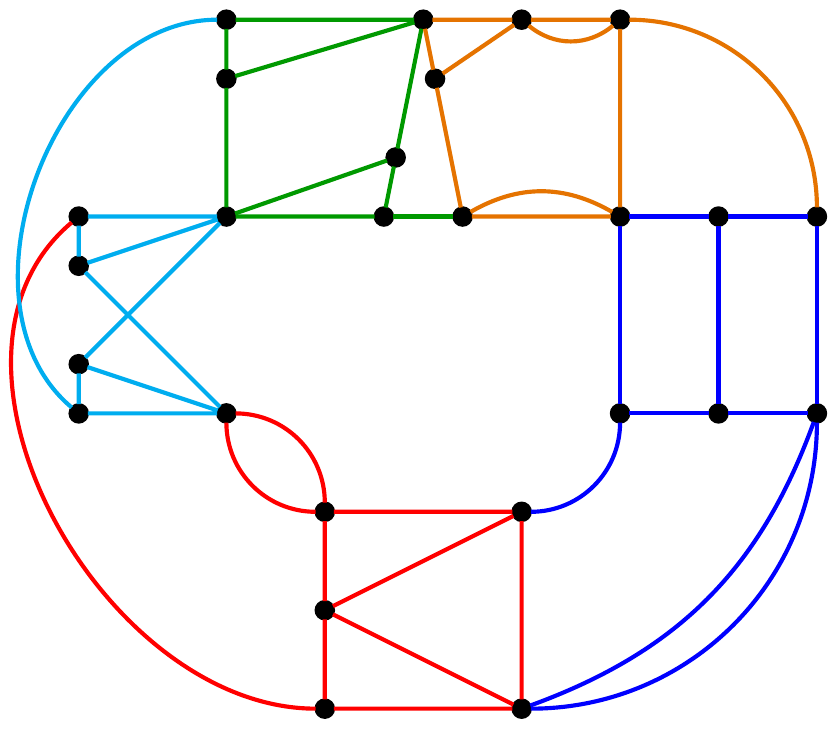}}
		\subcaption{A drawing showing $\scr(G)=2$ by twisting the left \emph{AAL}-tile (light blue), as schematized in \cref{TwistWithoutI}. \cref{TwistVIA} shows how redrawing e.g. the \emph{VIAdL}-tile would look like.}
	\end{subfigure}
	\caption{The example graph $G$ from \cref{fig:2cc_example} in the context of the simple crossing number.}
	\label{fig:simpledrawing}
\end{figure}

\begin{theorem}\label{thm:simple}
	Any large $2$-crossing-critical graph~$G$ has $\scr(G) = 2$.
\end{theorem}

\begin{proof}
	Since $2=\crn(G) \leq \scr(G)$, the claim follows if each $G$ admits a $1$-planar drawing with $2$ crossings.

	We achieve this by performing a \emph{twist operation} at a single elementary tile~$X$:
	In the natural drawing on the Möbius strip (cf.\ Figure~\ref{fig:2cc_example}) each tile is drawn planarly but we cannot identify the left-most with the right-most wall vertices in a planar fashion.
	Twisting a tile $X$ within this drawing means to invert the vertical order of its \change{left or} right wall vertices, thereby incurring some crossings within~$X$. Given this twisted tile, all subsequent tiles can be planarly drawn \removedS{by vertically inverting them}and we can now identify the left-most and right-most wall vertices planarly  (cf.\ Figure~\ref{fig:simpledrawing}). Thus we do not need any crossings except for those within~$X$; we will discuss them below.

	\begin{figure}
	\centering
	\begin{subfigure}[c]{\textwidth}
		\centering
		\winIgnore{\includegraphics[scale=.6]{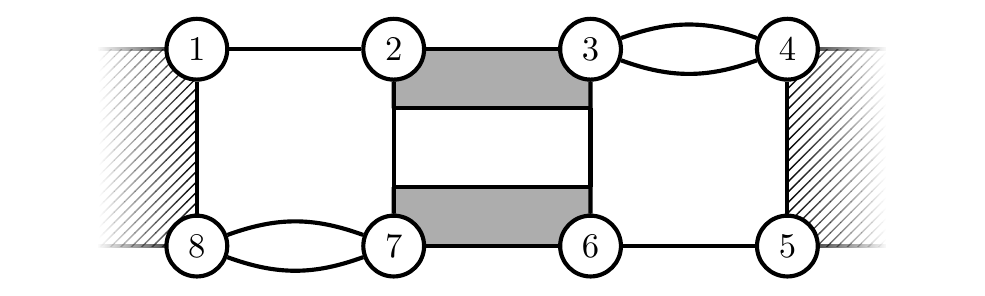}}%
		\clap{\raisebox{0.75cm}{$\Rightarrow$}}%
		\winIgnore{\includegraphics[scale=.6]{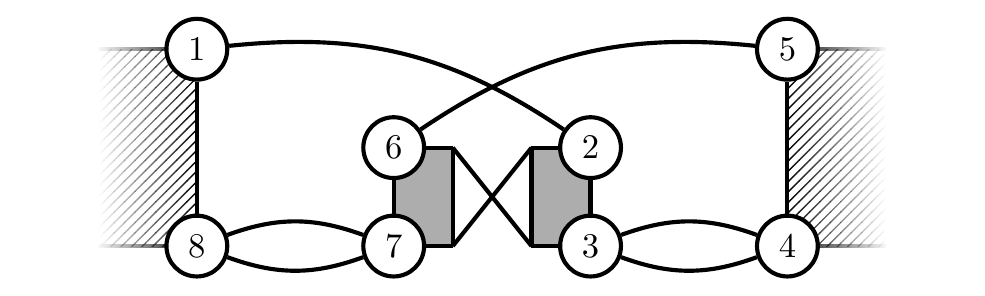}}
		\subcaption{Twisting a tile without $\pI$.}
		\label{TwistWithoutI}
	\end{subfigure}
	\begin{subfigure}[c]{\textwidth}
		\centering
		\winIgnore{\includegraphics[scale=.6]{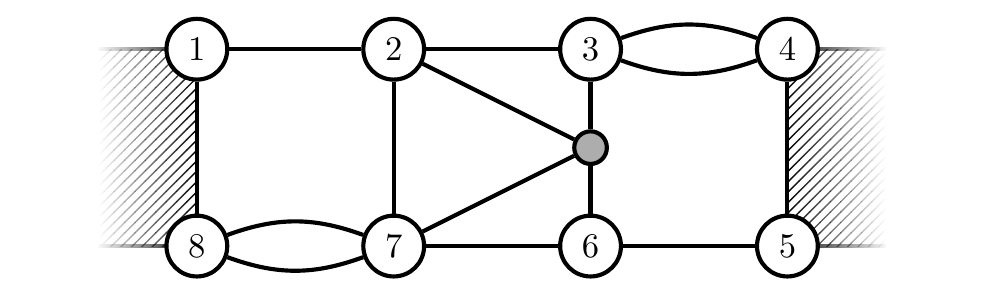}}%
		\clap{\raisebox{0.75cm}{$\Rightarrow$}}%
		\winIgnore{\includegraphics[scale=.6]{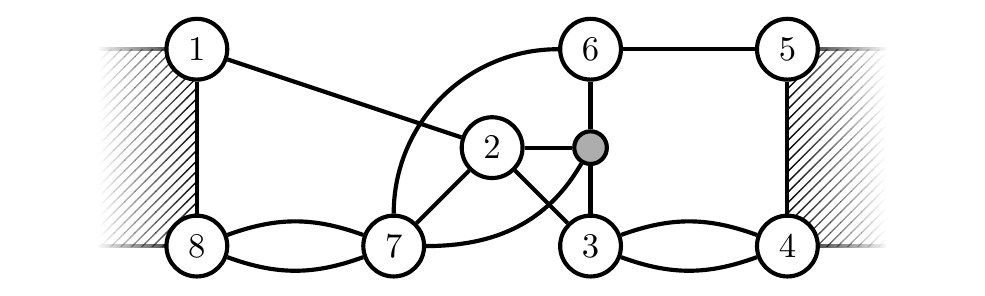}}
		\subcaption{Twisting a $\pV\pI\pA\fdL$-tile.}
		\label{TwistVIA}
	\end{subfigure}
	\begin{subfigure}[c]{\textwidth}
	\centering
	\winIgnore{\includegraphics[scale=.6]{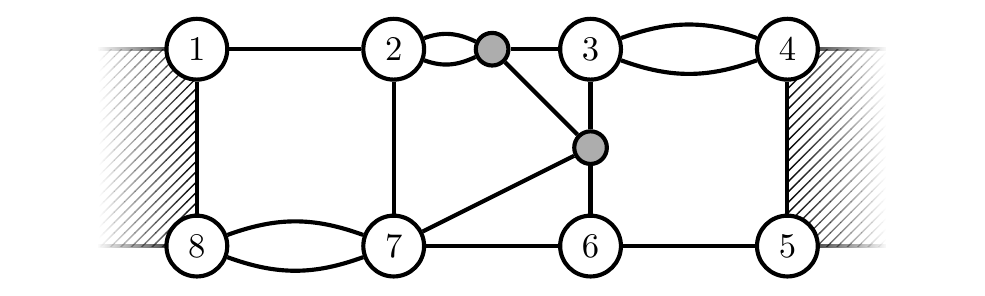}}%
	\clap{\raisebox{0.75cm}{$\Rightarrow$}}%
	\winIgnore{\includegraphics[scale=.6]{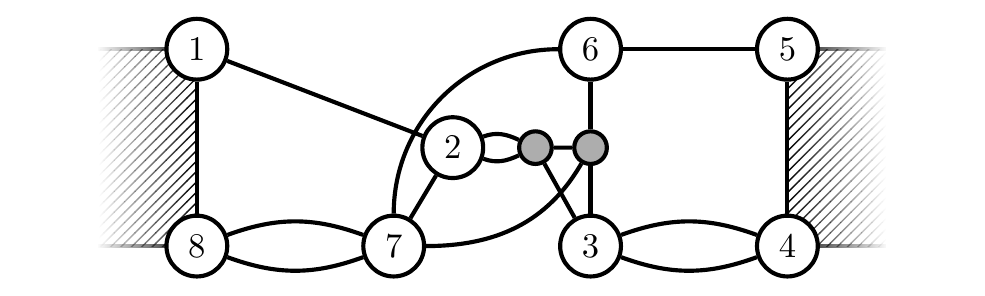}}
	\subcaption{Twisting a $\pB\pI\pA\fdL$-tile.}
	\label{TwistBIA}
	\end{subfigure}
	\caption{Twisting a tile.}
	\label{Twist}
	\end{figure}

	First, we consider a twisted tile consisting of a $\fdL$-frame and a picture without $\pI$.
	\Cref{TwistWithoutI} gives an abstract sketch (as well as its twist) of such tiles where the gray area hides crossing-free picture details.
	Hence, the twisting of these tiles can be drawn $1$-planarly with $2$-crossings.

	Next, we prove the claim for any twisted tile consisting of a $\fdL$-frame and a picture with identification.
	\removed{$1$-planar drawings of twisted tiles containing a  $\pV \pI \pA$-picture are given in Figure 5b.
	It is easy to see that frames containing a picture with a $\pB$ and an identification can be drawn in a similar way as $\pV \pI \pA \fdL$ by placing the end vertices of the double edge in $\pB$ paths close together.}\change{To this end, recall that there are only four pictures with identification: $\pV \pI \pA$, $\pB \pI \pA$, $\pA \pI \pV$, and $\pA \pI \pB$.
	We give $1$-planar drawings of a twisted $\pV \pI \pA \fdL$-tile and  a twisted $\pB \pI \pA \fdL$-tile in \Cref{TwistVIA,TwistBIA}.
		The solutions for $\pA \pI \pV \fdL$- and $\pA \pI \pB \fdL$-tiles are identical up to mirroring.}
	\removedS{Note that}\change{C}ontracting the double edges $(3,4)$ and/or $(7,8)$ in the given drawings maintains $1$-planarity and the simple crossing number.
	Hence, the given drawings can be transformed to tiles with an $\fL$-frame.
\end{proof}

\section{Chromatic Number}
\label{sc:chromaticnum}
The question whether four colors are sufficient to color a map (in the sense of a separation of the plane into contiguous regions) such that no two adjacent regions (e.g., countries in a visual representation of national territories) are colored the same, plagued mathematicians and many other researchers since the late 19th century.
It was finally, but not uncontroversially, positively answered in 1976 facilitating a computer-assisted proof~\cite{AH77}.
Consequently, any planar graph has a vertex coloring that uses at most four colors.
For general graphs however, it is NP-hard to decide whether a given number~$k \geq 3$ of colors suffices (the smallest such $k$ is the graph's \emph{chromatic number})~\cite{GJ79} and even constant-factor approximations in polynomial time are impossible (unless P$=$NP)~\cite{Zuc07}.
The chromatic number of graphs is of interest in applications like scheduling, register allocation, and pattern matching
\cite{Cha82, Mar04, Lew16}. Ringel proved that $1$-planar graphs can be colored using at most seven colors \cite{Rin65}.

In this section, we study the chromatic number of \ltwocc{}s. We start by a characterization of bipartite, i.e., $2$-colorable, such graphs in Theorem~\ref{thm:CharacterizationBipartite} and proceed to improve on Ringel's result by proving that each \ltwocc is $4$-colorable, cf.~Theorem~\ref{thm:4Colorability}.
To show that this bound is tight (at least in some cases), we present an infinite family of \ltwocc{s} that are not $3$-colorable.
Finally, this is complemented by an infinite family of \ltwocc{}s with chromatic number~$3$.

\begin{definition}
		A \emph{(vertex) coloring} of a graph~$G$ is a function~$c\colon V(G) \rightarrow \mathbb{N}^+$, such that $c(v) \neq c(w)$ for every edge $vw \in E(G)$.
		Graph~$G$ is \emph{$k$-colorable} if it admits a coloring using at most $k$ colors and we call such a coloring a \emph{$k$-coloring}.
        The \emph{chromatic number} of~$G$ is the smallest $k$ such that $G$ is $k$-colorable.
\end{definition}

\begin{figure}
	\centering
	\begin{subfigure}[c]{\textwidth}
		\centering
		\winIgnore{\includegraphics[width=\textwidth]{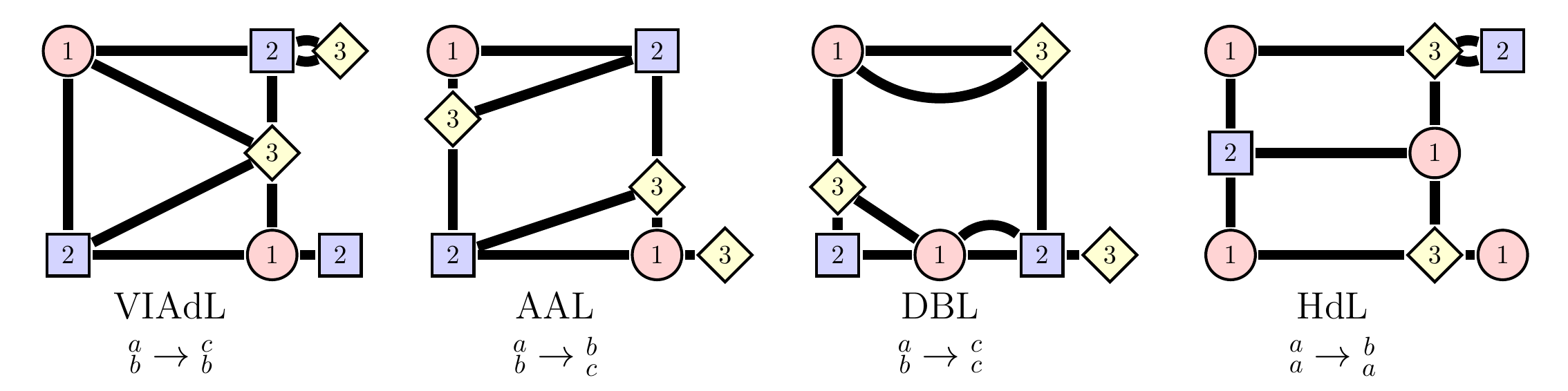}}
		\caption{Individual \change{$3$-}colorings of each elementary tile in the graph. While they are not compatible to each other as is, they can be thought of as propagations, see text.\label{fig:nodecol-example-pre}
	}
	\end{subfigure}
	\\
	\begin{subfigure}[c]{\textwidth}
		\centering
		\winIgnore{\includegraphics[scale=0.4]{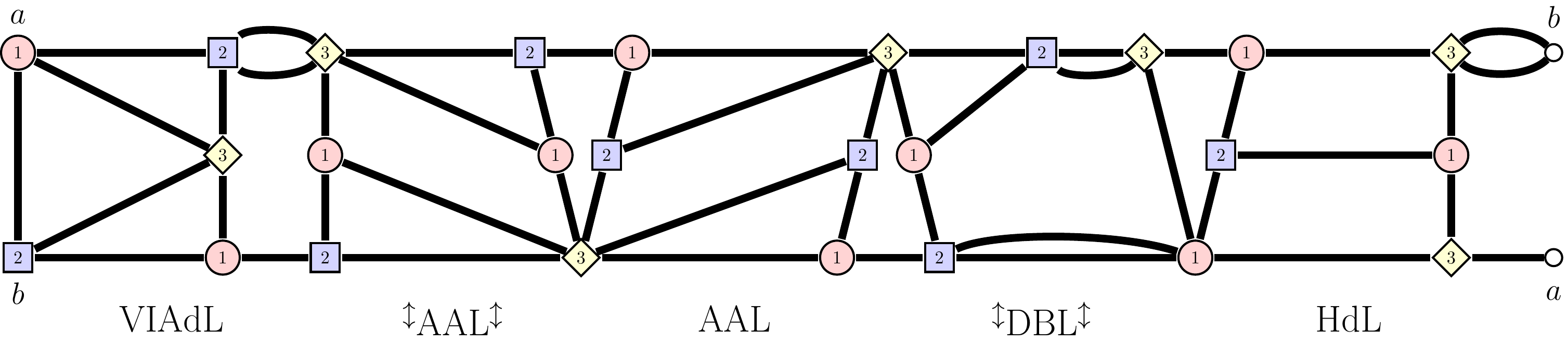}}
		\caption{A \removedS{$4$}\change{$3$}-coloring of the example graph from \cref{fig:2cc_example}, based on the above propagations.\label{fig:nodecol-example-post}}
	\end{subfigure}
	\caption{Using propagations to color a graph.}
	\label{fig:nodecol-example}
\end{figure}

In \cref{fig:nodecol-example-pre} we can see that each elementary tile of the example graph from \cref{fig:2cc_example} can be colored with at most \removedS{$4$}\change{$3$} colors on its own.
But if we were to use these exact colorings in the full graph, the colorings of the wall nodes would clash.
In the following, while we formally construct the graphs by joining tiles whose right wall is inverted (see Definition~\ref{def:constructionLarge2ccGraphs}), it will be helpful to take the viewpoint already discussed following that definition, that we may vertically invert every second tile completely. This allows us to (mentally and in the figures) visualize each tile planarly.

We view every coloring of \cref{fig:nodecol-example-pre} as a \emph{$2$-propagation} (to be formally defined below), and substitute explicit colors as needed.
Intuitively (cf.\ \cref{fig:nodecol-example-post}), start with coloring the \emph{VIAdL}-tile as proposed by its individual coloring. As its right wall vertices (with now fixed colors) are the left wall vertices of the neighboring \emph{AAL}-tile,
we cannot directly color the latter tile as desired by its individual coloring. Let $x_1,x_2$ and $y_1,y_2$ denote the two left and right wall vertices of this \emph{AAL}-tile (observe that it is drawn as ${}^{\updownarrow}\!\emph{AAL}^{\updownarrow}$ in~$G$), respectively. We are only interested in the following properties of the \emph{AAL}-coloring in \cref{fig:nodecol-example-pre}: $c(y_1)$ is distinct from $c(x_1)$ and $c(x_2)$, and $c(y_2) \removed{\neq}\change{=} c(x_2)$. This allows us to substitute the color classes within this tile accordingly and proceed with the next tile. Put chiefly, a $2$-propagation is the notion that, given a coloring of its left wall vertices, we know about the existence of a tile-coloring yielding certain coloring-properties on its right wall vertices.
This concept can be formalized as follows:
\begin{definition}
	Let $T=(G,x,y)$ be a \change{2-}tile.
Consider a vertex coloring~$c$ of $G$. The colors of $x_1,x_2$ ($y_1,y_2$) are the \emph{input colors} (\emph{output colors}, respectively; each in that order) of $T$.
Two $k$-colorings~$c,c'$ of~$T$ are \emph{equivalent},
if $(c(v) = c(w)) \!\Leftrightarrow\! (c'(v) = c'(w))$ for each pair of wall vertices~$v,w \in \{x_1,x_2,y_1,y_2\}$.
We call the induced equivalence classes \emph{(vertex-coloring-)propagations} and denote them by \[
	\substack{c(x_1)\\c(x_2)}\rightsquigarrow\substack{c(y_1)\\c(y_2)},
\]
using some representative coloring~$c$.
We may use the term \emph{$k$-propagation} to specify that $c$ is a $k$-coloring.
\end{definition}
To aid comprehensibility, whenever we state propagations, we will denote each color by a unique letter from $\{a,b,c,d\}$ instead of a number.
Observe that, while elementary tiles will yield the base cases of our propagations, the joins of tiles yield again tiles; therefore we may naturally concatenate propagations when joining two tiles, requiring only simple color substitutions in the second tile. Consider the first two tiles in the example graph \cref{fig:nodecol-example-post}: their propagations
$\substack{a\\b}\rightsquigarrow\substack{c\\d}$ and $\substack{a\\b}\rightsquigarrow\substack{c\\d}$ lead to
$\substack{a\\b}\rightsquigarrow\substack{c\\d}\rightsquigarrow\substack{a\\b}$ and thus the propagation $\substack{a\\b}\rightsquigarrow\substack{a\\b}$ over the first two tiles.

We first consider small odd cycles that arise already when joining two elementary tiles.

\begin{lemma}\label{BuildingBlocksBipartite}
	Let $\sig(G)=\sig(T_0) \ldots \sig(T_{2m})$ be a signature of a \ltwocc $G$. Then $G$ is not bipartite if $\sig(T_{2m})\sig(G)$
	contains an element of $\{\pD\pD\fdL\,\pH, \pD\pD\fL\,\pD\pD, \pH\fdL\,\pD\pD, \pH\fL\,\pH\}$ as a substring.
\end{lemma}
\begin{proof}
The subgraph corresponding to $\pD\pD\fL\,\pD\pD$ contains a triangle and the subgraphs corresponding to $\pH\fL\,\pH$, $\pD\pD\fdL\,\pH$, and $\pH\fdL\,\pD\pD$ each contain a $5$-cycle.
\end{proof}

\change{Recall, that by Observation~\ref{lemma:picture_triangle} every tile whose signature contains \pA, \pV, or \pB~has a triangle and is therefore not bipartite.}

Let us now consider the last ``global'' tile~$T$ whose cyclization yields a \ltwocc. We proceed to show that it is bipartite if none of the above local obstructions are present.
Note that in the following lemma, we do \emph{not} consider the final cyclization just yet.

\begin{lemma}\label{LemmaYieldingBipartite}
Let $T=T_0^{\updownarrow}\otimes T_1^{\updownarrow} \otimes \ldots\otimes T_{2m}^{\updownarrow}$ where $T_i \in \mathcal S$. The tile $T$ is bipartite if $\sig(T_i)\sig(T_{(i+1) \bmod (2m+1)})$ starts with $\pD\pD\fdL\,\pD\pD$, $\pD\pD\fL\,\pH$, $\pH\fdL\,\pH$, or $\pH\fL\,\pD\pD$ for each $i \in \mathbb N$ with $0 \leq i \leq 2m$.
\end{lemma}

\begin{figure}
	\centering
	\winIgnore{\includegraphics[width=\textwidth]{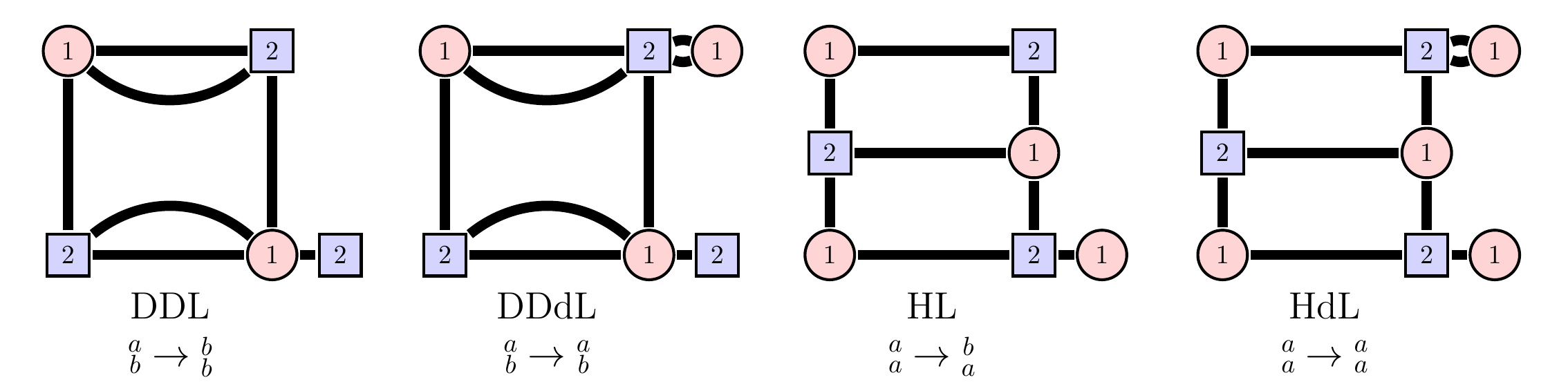}}
	\caption{2-vertex-colorings $P_{\fL}^{\pD\pD}$ for $\pD\pD\fL$, $P_{\fdL}^{\pD\pD}$ for $\pD\pD\fdL$, $P_{\fL}^{\pH}$ for $\pH\fL$, and $P_{\fdL}^{\pH}$ for $\pH\fdL$.
	}
	\label{FigureNodeColorBipartite}
\end{figure}

\begin{proof}
\cref{FigureNodeColorBipartite} shows the existence of the following $2$-propagations:
\begin{align*}
	P_{\fL}^{\pD\pD} \coloneqq \substack{a\\b}\rightsquigarrow\substack{b\\b} &\text{, for \pD\pD\fL-tiles,}\\
	P_{\fdL}^{\pD\pD} \coloneqq \substack{a\\b}\rightsquigarrow\substack{a\\b} &\text{, for \pD\pD\fdL-tiles,}\\
	P_{\fL}^{\pH} \coloneqq \substack{a\\a}\rightsquigarrow\substack{b\\a} &\text{, for \pH\fL-tiles,}\\
	P_{\fdL}^{\pH} \coloneqq \substack{a\\a}\rightsquigarrow\substack{a\\a} &\text{, for \pH\fdL-tiles.}
\end{align*}
	By assumption each elementary tile is from $\{\pD\pD\fL, \pD\pD\fdL, \pH\fL, \pH\fdL\}$ and $\sig(T)$ is a subsequence of $((\pD\pD\fdL)^*\pD\pD\fL(\pH\fdL)^*\pH\fL)^+$ (using standard notation of regular expressions).

	First we will look only at the case where each elementary tile of $T$ is either $\pD\pD\fL$ or $\pH\fL$.
	Then $\sig(T)$ is a subsequence of $(\pD\pD\fL\,\pH\fL)^+$.
	Every subsequence $\pD\pD\fL\,\pH\fL$ of the latter admits the propagation $\substack{a\\b}\rightsquigarrow\substack{b\\b}\rightsquigarrow\substack{b\\a}$ by using $P_{\fL}^{\pH}$ and $P_{\fL}^{\pD\pD}$ (observe that the vertical order reverses in the second propagation as this tile is vertically inverted w.r.t.\ the first).
	Iterating these propagations yields a $2$-coloring.

	If we now also allow  $\pH\fdL$-tiles, we see that each maximal
	$(\pH\fdL)^+$-subsequence admits the propagation $\substack{a\\a}\rightsquigarrow\substack{a\\a}$ by repeatedly using $P_{\fdL}^{\pH}$.
	Therefore each subsequence $(\pH\fdL)^+\pH\fL$ admits the same overall propagation as an individual \pH\fL-tile.
	Similarly, any $(\pD\pD\fdL)^+$-subsequence admits the propagation $\substack{a\\b}\rightsquigarrow\substack{a\\b}$ by repeatedly using $P_{\fdL}^{\pH}$,
	and thus $(\pD\pD\fdL)^+\pD\pD\fL$ admits the same overall propagation as an individual \pD\pD\fL-tile.
	Thus $T$ is $2$-colorable.

\end{proof}

Finally, we can fully characterize bipartite \ltwocc{}s.
\begin{theorem}\label{thm:CharacterizationBipartite}
	A \ltwocc~$G$ is 2-colorable if and only if its signature can be written as $\sig(G)=\sig(T_0)\ldots\sig(T_{2m})$, where $T_i$ is an elementary tile for $0\leq i \leq 2m$, such that:
	\begin{enumerate}[(i)]
		\item tile~$T_0$ contains an \pH-picture \change{(defined in  Figure \ref{fig:tilePictures})}, and
		\item each $\sig(T_i)\sig(T_{(i+1) \bmod (2m+1)})$ starts with $\pD\pD\fdL\,\pD\pD$, $\pD\pD\fL\,\pH$, $\pH\fdL\,\pH$, or $\pH\fL\,\pD\pD$ for $0 \leq i \leq 2m+1$; and
		\item the number of $\fL$ frames in $\{T_{2i}\}_{0 \leq i \leq m}$ is odd.
	\end{enumerate}
\end{theorem}
\begin{proof}
	By Observation \ref{lemma:picture_triangle} and Lemma \ref{BuildingBlocksBipartite} each bipartite \ltwocc satisfies (ii).
	Moreover, since $(\pD\pD\fd\fL)^k$ %
	for odd $k$ is not bipartite, each bipartite \ltwocc contains an $H$-picture, implying (i).

	Hence, it only remains to prove that a \ltwocc satisfying (i) and (ii) is bipartite if and only if it also satisfies (iii). To this end,
	assume~$G$ satisfies (i) and (ii).

	By Lemma \ref{LemmaYieldingBipartite}, $(G',\change{\langle x_1, x_2\rangle,\langle y_2, y_1\rangle}) \coloneqq T_0^{\updownarrow}\otimes T_1^{\updownarrow} \otimes \ldots\otimes T_{2m}^{\updownarrow}$ admits a $2$-coloring $c$.
	Since such a coloring is unique (up to isomorphism and relabeling of colors), $G$ is bipartite if and only if $c$ induces a $2$-coloring on~$G$, i.e.\ $c(x_1) = c(y_1)$ and $c(x_2) = c(y_2)$.

	As $T_0$ contains an \pH-picture $c(x_1) = c(x_2)$ and as furthermore $G$ is a \ltwocc $c(y_1) = c(y_2)$, $c$ induces a $2$-coloring on $G$ if and only if $c(x_1) = c(y_1)$.
\removed{We consider two cases:\\[1em]
		1. If $T_{2m}$ contains an $\pH$-picture, we have $T_{2m}=HdL$ by Lemma~\ref{BuildingBlocksBipartite.}
			Thus, $c(y_1) = c(y_2)$ $c$ induces a $2$-coloring on $G$ if and only if $c(x_1) = c(y_2)$ and $c(x_2) =  c(y_1)$ which is equivalent to $c(x_1) \neq c(y_1)$.
			We observe that $x_1$ has odd distance from $y_1$ in the unclosed global tile if and only if property (iii) holds.\\[1em]
	2. If $T_{2m}$ contains a $\pD$-picture, we have $T_{2m}=DDL$ by Lemma~\ref{BuildingBlocksBipartite.}
			Thus, $c$ induces a $2$-coloring on $G$ if and only if $c(x_1) \neq c(y_2)$ and $c(x_2) \neq c(y_1)$ which is equivalent to $c(x_1) = c(y_1)$.
			We observe that opposed to the first case, $x_1$ has even distance from $y_1$ if and only if property (iii) holds.\\[1em]
	}\change{
	To this end we look at the parity of a path between $x_1$ and $y_1$.
	As $G'$ is bipartite every such path has the same parity.
	Our path consists of the direct path between the (non-inverted) top wall nodes for tiles $T_{2i}$ for $0\leq i \leq m$ and the direct path between the (non-inverted) bottom wall nodes of tiles $T_{2i+1}$ for $0\leq i < m$.
	We notice, that the (edgewise-)distance between the bottom wall nodes of each of our four elementary tiles is always $2$.
	The same is the case for the distance between the top wall nodes for \fdL-framed elementary tiles but for \fL-framed elementary tiles the distance between the top wall nodes is $1$.
	Therefore if and only if property (iii) holds, the distance between $x_1$ and $y_1$ is even and therefore $c$ induces a $2$-coloring on $G$.
}
\end{proof}

In \cref{fig:nodecol-example}, we have already seen a $4$-coloring of the example graph.
Let us now generalize this way of coloring to show that, like any planar graph, indeed every \ltwocc requires at most $4$ colors.

\begin{figure}
	\centering
	\begin{minipage}[t]{.55\textwidth}
		\centering
		\winIgnore{\includegraphics[scale=.55]{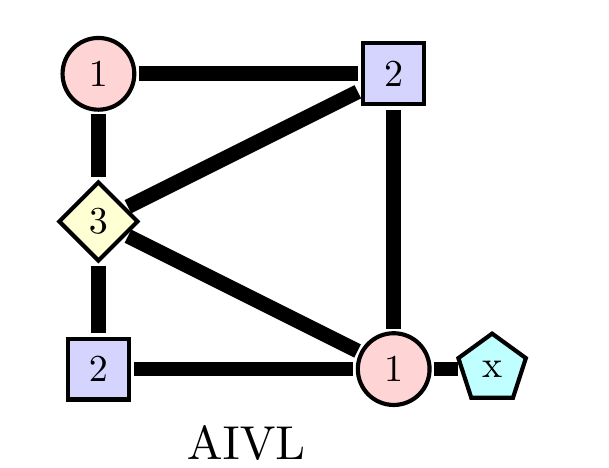}}
		\caption{The unique \mbox{$3$-coloring} of an \pA\pI\pV\fL-tile (up to color-substitution and the free choice for `x'$\neq1$).}
		\label{fig:nodecol-AIVL}
	\end{minipage}
	\hfill
	\begin{minipage}[t]{.35\textwidth}
		\centering
		\winIgnore{\includegraphics[scale=.55]{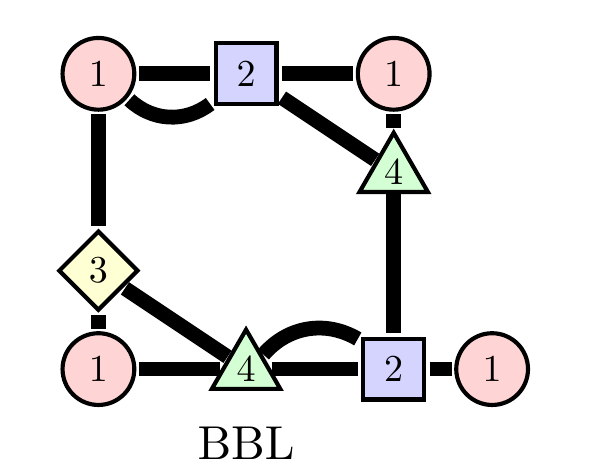}}
		\caption{A $3$-propagation $\protect\substack{a\\a}\rightsquigarrow\protect\substack{a\\a}$ of a \pB\pB\fL-tile.}
		\label{fig:nodecol-BBL}
	\end{minipage}
\end{figure}

\begin{theorem}\label{thm:4Colorability}
	Every \ltwocc is $4$-colorable.
\end{theorem}

\begin{proof}
	It is easy to verify that every elementary tile admits a $4$-propagation $\substack{a\\b}\rightsquigarrow\substack{c\\d}$; we list all these propagations explicitly
	in \Cref{Figures4Colorable1,Figures4Colorable2} in the appendix. Thus, any tile that consists of two joined elementary tiles,
	admits the $4$-propagation $\substack{a\\b}\rightsquigarrow\substack{a\\b}$.
	We \removedS{may}use $\substack{a\\b}\rightsquigarrow\substack{a\\b}$ on all but $3$ consecutive elementary tiles and show that the join of these $3$ tiles has
	an $\substack{a\\b}\rightsquigarrow\substack{b\\a}$ propagation:
	As also shown in \Cref{Figures4Colorable1,Figures4Colorable2}, each elementary tile also admits the $4$-propagation $\substack{a\\b}\rightsquigarrow\substack{c\\b}$. Thus, three such tiles admit the required $4$-propagation (recall that the tile in the middle is drawn vertically inverted w.r.t.\ the other two):
	\[		\substack{a\\b} \rightsquigarrow \substack{c\\b} \rightsquigarrow \substack{c\\a} \rightsquigarrow \substack{b\\a}.\]
\end{proof}

Next, we present a class of large $2$-crossing-critical graphs \removedS{which}\change{that} are not $3$-colorable. This shows that the bound presented above is tight for an infinite number of cases.

\begin{obs}
	\removedS{Each}\change{Every} \ltwocc~\removedS{with}\change{where} every elementary tile \change{is an} \pA\pI\pV\fL\change{-tile} is not $3$-colorable.
\end{obs}
\begin{proof}
	With \cref{fig:nodecol-AIVL} it is straightforward to verify that each $3$-propagation of an $\pA\pI\pV\fL$-tile is either $P_{1} = \substack{a\\b}\rightsquigarrow\substack{b\\c}$ or $P_{2} = \substack{a\\b}\rightsquigarrow\substack{b\\b}$.
	Since the two vertices on the left wall of \change{an} \pA\pI\pV\fL-tile have to be colored differently,
	each elementary tile uses propagation $P_1$.
	Thus, any join of an even number of elementary tiles propagates $\substack{a\\b}\rightsquigarrow\substack{a\\b}$.
	But then the last tile would have to propagate $\substack{a\\b}\rightsquigarrow\substack{b\\a} \neq P_1$.
\end{proof}

Complementing this, there are also infinitely many \ltwocc{s} with chromatic number~$3$.

\begin{obs}
	\removedS{Each}\change{Every} \ltwocc $G$ \removedS{with}\change{where} every elementary tile \change{is a} \pB\pB\fL\change{-tile} has chromatic number $3$.
\end{obs}
\begin{proof}
	By Theorem \ref{thm:CharacterizationBipartite}, $G$ is not bipartite.
	On the other hand, \cref{fig:nodecol-BBL} shows that \change{a} \pB\pB\fL\change{-tile} admits a $3$-propagation $\substack{a\\a}\rightsquigarrow\substack{a\\a}$ that we may use on all tiles.
\end{proof}

The previous observations point to the following open problem:
\begin{question}\label{q:character3}
What is the full characterization of $3$-colorable \ltwocc{}s?
\end{question}

Although a graph-theoretic characterization of $3$-colorable \ltwocc{}s is an open question, we can efficiently decide $3$-colorability algorithmically:
\begin{obs}
	Using the fact that \ltwocc{s} have bo\-un\-d\-ed treewidth (see Corollary~\ref{cor:bounded_treewidth} below), Courcelle's theorem~\cite{Cou90} yields a linear-time algorithm to decide whether
	a given \ltwocc{} is 3-colorable.
\end{obs}

\section{Chromatic Index}
\label{sc:chromaticin}

In this section, we investigate the chromatic index of \ltwocc{s}. The chromatic index is the minimum number of colors necessary to color edges of a graph, such that no two edges incident to the same vertex share a color. A trivial lower bound for the chromatic index is the maximum degree of the graph. Determining the chromatic index of a general graph is NP-hard~\cite{Hol81}. However, there are classes of graphs for which the chromatic index can be shown to be close to the trivial lower bound. Simple graphs are said to be \emph{class 1} if their chromatic index equals the maximum degree, and \emph{class 2} otherwise, see e.g.~\cite{Cao+19}.
	However, the situation is more complicated for \removed{multi}graphs \change{that are not simple}, as the graph's \emph{density} (i.e., maximum ratio between the number of edges and vertices, over all induced subgraphs) is also a natural lower bound for the \removed{multi}graph's chromatic index. This motivates the following slightly different definition~\cite{Cao+19}: a \removed{multi}graph is \emph{first class} when its chromatic index matches the lower bound given by the maximum degree or the density, and \emph{second class} otherwise.

	By construction, the density of \ltwocc{}s is low.
In fact, we show that all \ltwocc{s} are first class by showing that
they require only as many edge colors as their maximum degree.
To this end, we exhibit such edge colorings for elementary tiles and combine them to a coloring of the full graph.

\begin{definition}
	An \emph{edge coloring} of a graph~$G$\removed{~= (V,E)} is a function $c\colon E\change{(G)}\rightarrow\mathbb N^+$ such
	that $c(e) \neq c(f)$ for each pair~$e,f\in E\change{(G)}$ of adjacent edges.
	A \emph{$k$-edge-coloring} is an edge coloring that uses at most $k$ colors.
	The \emph{chromatic index} of~$G$ is the smallest $k$ such that a $k$-edge-coloring of $G$ exists.
	In particular, if $G$ admits a $\Delta(G)$-edge-coloring, $G$ is said to be \emph{first class}.
\end{definition}

Similarly to our findings on chromatic numbers, we will use color propagations to investigate edge colorings. An example can be seen in \cref{fig:edgecol-example}. \cref{fig:edgecol-example-pre} shows five edge color propagations, one for each tile. Consider two neighboring tiles $T_1,T_2$ ($T_1$ left of $T_2$).
	For edge color propagations, the edges in $T_1$ incident to $T_1$'s right wall are
	of interest, as they form restrictions for the edges in $T_2$ that are incident to $T_2$'s left wall vertices.
	Thus, when showing a propagation for tile $T_2$, we also need to show these incident $T_1$-edges (the \emph{input edges} \change{of $T_2$}), to the left of the wall vertices. The $T_2$-edges incident to $T_2$'s right wall form the \emph{output edges} \change{of $T_2$}.
	By substituting colors, we can now again assign these propagations to a list of joined tiles such that their colors match. Note that in our example graph we have two distinct color propagations for the two $\pA\pA\fL$-tiles, since in the full graph (\cref{fig:edgecol-example-post}) their left wall vertex $y_2$ becomes a vertex of degree $6$ in the first, and degree $5$ in the second case.

\begin{figure}[tb]
	\centering
	\begin{subfigure}[c]{\textwidth}
		\centering
		\winIgnore{\includegraphics[width=\textwidth]{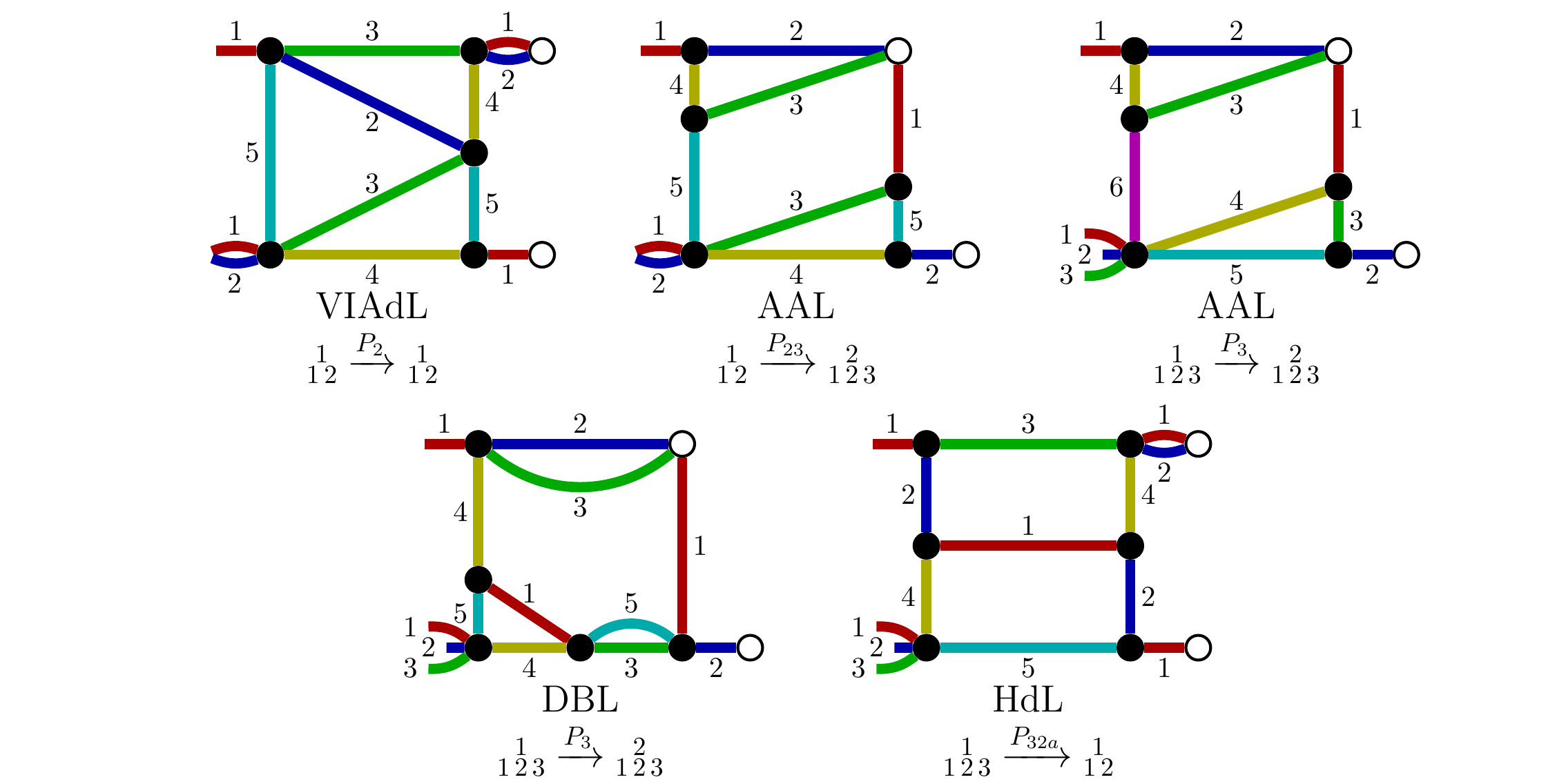}}
		\caption{The raw propagations.\label{fig:edgecol-example-pre}}
	\end{subfigure}
	\\
	\begin{subfigure}[c]{\textwidth}
		\centering
		\winIgnore{\includegraphics[width=\textwidth]{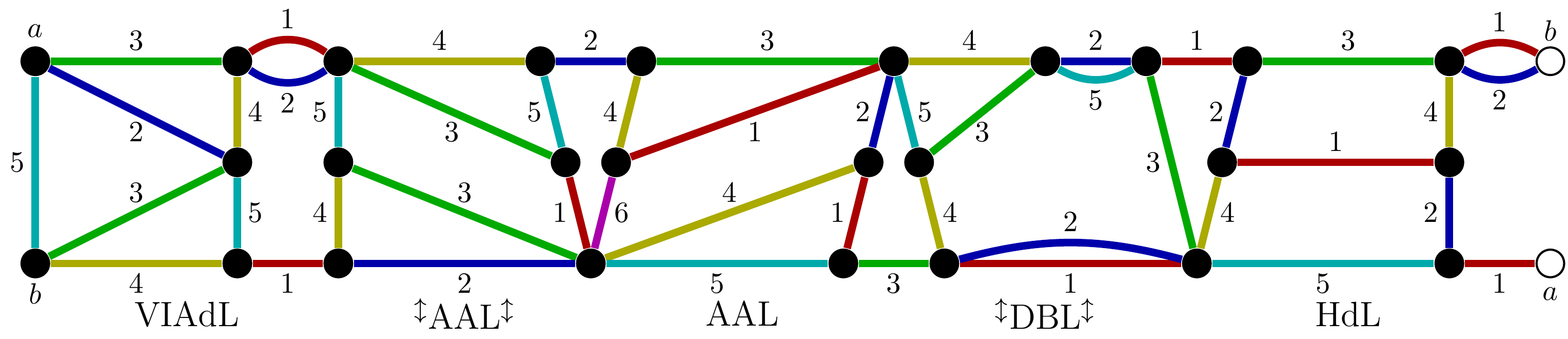}}
		\caption{A full edge coloring of maximum vertex degree of the example graph.\label{fig:edgecol-example-post}}
	\end{subfigure}
	\caption{Edge color propagations in use.}
	\label{fig:edgecol-example}
\end{figure}

Let $T=(G,x,y)$ be a tile.
The edges of $T$ that are incident to vertices on $T$'s right wall are the \emph{output edges} of $T$.
We observe that each tile~$T$ of a \ltwocc has either three or four output edges,
where all but one edge~$e$ are pairwise adjacent.
We call \change{this edge} $e$\change{, which is the unique edge incident to the degree-$1$ vertex of the frame,} the \emph{single edge} of $T$.
Consider an edge coloring of $G$,
the colors of the output edges of $T$ are its \emph{output colors}.
We denote output colors by $\substack{a\\b\,c\,d}$, where $a$ refers to the color of
the single edge and $b,c,d$ are the colors of the three adjacent edges (in no particular order).
For those tiles that have only two instead of three adjacent such edges, we instead write~$\substack{a\\b\,c}$.
For a cyclic sequence of tiles that contains~$T$, the output edges (output colors) of
$T$'s predecessor are the \emph{input edges} (\emph{input colors}, respectively) of $T$.
We employ the same notation for input and output colors.

\begin{definition}
	Two $k$-edge-colorings~$c,c'$ of a tile~$T$ and its input edges are \emph{equivalent}
	if $c(e) = c(f) \iff c'(e) = c'(f)$ for each pair~$e,f \in X$, where $X$ is the set of input and output edges of $T$.
	We call the induced equivalence classes \emph{(edge-coloring-$k$-)propagations} and denote them by
	$\mathcal A \longrightarrow \mathcal B$, where $\mathcal A$ ($\mathcal B$) are the input colors (output colors, respectively) of $T$.
\end{definition}

A tile (in a sequence) that has $k+1$ input and $\ell+1$ output colors, is a \emph{$k$-$\ell$-tile}.
We observe that any \ltwocc is a cyclic sequence of elementary $2$-$2$-, $2$-$3$-, $3$-$3$- and $3$-$2$-tiles.
Given a tile $T$ in a \ltwocc~$G$, we denote its \emph{maximum degree} in~$G$ by
$\Delta_G(T) = \max_{v \in V(T)}|\{e \in E(G) \mid v \in e\}|$.

We know from Observation~\ref{obs:maxdeg} that the maximum degree of a \ltwocc is $4\leq \Delta(G) \leq 6$. In order to show that all \ltwocc{}s are first class, we will first restrict ourselves to those with $\Delta(G) \geq 5$; thereafter, we will also consider the case $\Delta(G) = 4$.
\begin{lemma}
	\label{lem:edgeColorPropagation}
	Consider a (cyclic) sequence~$Q$ of elementary tiles corresponding to a \ltwocc~$G$ with maximum degree~$\Delta(G)\geq 5$.
	Each $2$-$2$-tile of $Q$ admits the following propagation that uses $5$ colors:
	\begin{itemize}
		\item[] $P_2 \coloneqq \substack{1\\1\,2}\longrightarrow\substack{1\\1\,2}$.
	\end{itemize}
	Each other tile~$T$ of $Q$ has the below propagations, using at most $\Delta_G(T)$ colors:
	\begin{itemize}
		\item[]
			$P_{23} \coloneqq \substack{1\\1\,2}\longrightarrow\substack{2\\1\,2\,3}$ for $2$-$3$-tiles,\quad
			$P_3 \coloneqq \substack{1\\1\,2\,3}\longrightarrow\substack{2\\1\,2\,3}$ for $3$-$3$-tiles,
		\item[]
			$P_{32b} \coloneqq \substack{1\\1\,2\,3}\longrightarrow\substack{2\\1\,2}$ for $3$-$2$-tiles,\quad
			$P_{32a} \coloneqq \substack{1\\1\,2\,3}\longrightarrow\substack{1\\1\,2}$ for $3$-$2$-tiles.
	\end{itemize}
\end{lemma}
\begin{proof}
	This can (easily but tediously) be shown by demonstrating corresponding colorings for each elementary tile.
	\Cref{fig:ec_p2,fig:ec_p3,fig:ec_p23,fig:ec_p32a,fig:ec_p32b} list all cases in the appendix.
\end{proof}

\change{Note that each elementary tile admits several propagations.}
In the example graph of \cref{fig:edgecol-example}, \removedS{we differentiate the two occurrences of $\pA\pA\fL$}\change{there are two occurrences of $\pA\pA\fL$. They differ in that their left wall vertex $y_2$ has a degree of $5$ or $6$. We differentiate them} by referring to the first one as a $2$-$3$-tile, and the second as a $3$-$3$-tile.by referring to the first one as a $2$-$3$-tile, and the second as a $3$-$3$-tile. The full sequence of propagations used is $\substack{1\\1\,2}\xrightarrow{P_{2}}\substack{1\\1\,2}\xrightarrow{P_{23}}\substack{2\\1\,2\,3}\xrightarrow{P_{3}}\substack{3\\1\,2\,3}\xrightarrow{P_{3}}\substack{1\\1\,2\,3}\xrightarrow{P_{32a}}\substack{1\\1\,2}$. This coloring uses $\Delta(G)=6$ colors.

We will now use these propagations to obtain an edge coloring of arbitrary \ltwocc{}s with $\Delta(G)\geq 5$ and show that they are indeed first class.

\begin{lemma}\label{lemma:chromaticindex}
	\Ltwocc{}s $G$ with $\Delta(G)\geq 5$ are first class.
\end{lemma}

\begin{proof}
	Throughout this proof, we only consider propagations using $5$ colors for elementary $2$-$2$ tiles and propagations using at most $\Delta(T)$ colors for each other elementary tile $T$.

	First, assume $G$ does not decompose into elementary $3$-$3$-tiles only. Then, we prove the claim by decomposing $G$ into (not necessarily elementary) tiles admitting a $\substack{1\\1\,2}\longrightarrow\substack{1\\1\,2}$ propagation.
	We decompose $G$ into elementary $2$-$2$-tiles (which allow these via $P_2$) and tiles of the form $T=\otimes(T_0,\ldots,T_k)$ where $T_0$ is a $2$-$3$-tile, $T_k$ is a $3$-$2$-tile, $T_1, \ldots,T_{k-1}$ are $3$-$3$-tiles, and each $T_i$ is elementary.
 We only have to show that such a tile $T$ admits a $\substack{1\\1\,2}\longrightarrow\substack{1\\1\,2}$-propagation.

Iteratively applying $P_3$ to $T_1 \ldots T_{k-1}$ yields a $P_o\coloneqq\substack{1\\1\,2\,3}\longrightarrow\substack{2\\1\,2\,3}$-propagation if $k$ is even and a $P_{e}\coloneqq\substack{1\\1\,2\,3}\longrightarrow\substack{1\\1\,2\,3}$-propagation otherwise. We obtain the following propagations for $T$:
\[
	\substack{1\\1\,2}\xrightarrow{P_{23}}\substack{2\\1\,2\,3}\xrightarrow{P_{o}} \substack{1\\1\,2\,3} \xrightarrow{P_{32a}}\substack{1\\1\,2}\quad \text{and} \quad
	\substack{1\\1\,2}\xrightarrow{P_{23}}\substack{2\\1\,2\,3}\xrightarrow{P_{e}} \substack{2\\1\,2\,3} \xrightarrow{P_{32b}}\substack{1\\1\,2} \text{, respectively}.
\]

	Next, assume that $G$ consists of elementary $3$-$3$-tiles only. Note that using $P_3$, two subsequent such tiles admit the propagation
	$P_{3}^2 \coloneqq \substack{1\\1\,2\,3}\longrightarrow\substack{2\\1\,2\,3}\longrightarrow\substack{1\\1\,2\,3};$
	and three subsequent such tiles admit the propagation
	$P_{3}^3 \coloneqq \substack{1\\1\,2\,3}\longrightarrow\substack{2\\1\,2\,3}\longrightarrow\substack{3\\1\,2\,3}\longrightarrow\substack{1\\1\,2\,3}.$
	Since there is an odd number of elementary tiles, we can use $P_{3}^3$ for three subsequent elementary tiles and $P_{3}^2$ for the remaining pairs,
	obtaining a $\Delta(G)$-edge-coloring of $G$.
\end{proof}

Now that we have shown that \ltwocc{}s with $\Delta(G)\geq 5$ are first class, it remains to prove that those with $\Delta(G)=4$ are also first class.
\change{There is only a constant number of \twocc{s} with $3$ elementary tiles with potentially sporadic behavior; we are interested in the remaining infinite class.}

\begin{theorem}\label{thm:chromindex}
	\Ltwocc{s} $G$ with at least $5$ elementary tiles are first class.
\end{theorem}
\begin{proof}
	By Lemma~\ref{lemma:chromaticindex}, it remains to consider $\Delta(G) = 4$.
Let $G$ consist of the elementary tile $T_0, \ldots,$ $T_{2m}$ in this order. Throughout this proof, we only consider propagations using $4$ colors and assume $m \geq 2$.

We can find $4$-propagations for all tiles with maximum degree~$4$ (all corresponding colorings are depicted in the appendix), in particular they can be categorized as follows:
 \begin{align*}
	P_2\coloneqq\substack{1\\1\,2}\longrightarrow\substack{1\\1\,2}, \quad &\text{for $\pB\pV\fL$ and $\pV\pV\fdL$ (see \cref{fig:ec_p2}), and}\\
	P_w\coloneqq\substack{1\\1\,2}\longrightarrow\substack{3\\1\,3}, \quad &\text{for $\pB\pV\fL$ and $\pV\pV\fdL$ (see \cref{fig:ec_fourcolors}), and }\\
	P_s\coloneqq\substack{1\\1\,2}\longrightarrow\substack{1\\1\,3}, \quad &\text{for all other tiles of maximum degree $4$ (see \cref{fig:ec_fourcolors}).}
 \end{align*}
 We call the tiles $\pB\pV\fL$ and $\pV\pV\fdL$ \emph{whimsical}, while the others are \emph{sincere}.
 Let us prove the theorem by constructing a $\substack{1\\1\,2}\longrightarrow\substack{1\\1\,2}$-propagation for \removedS{$G'=\otimes(T_0, \ldots,$ $T_{2m})$}\change{$G$}. To this end, we consider the following three cases:

 {\it Case 1:} Assume there is an odd number of whimsical tiles. Then the number of sincere tiles is even.
 We obtain the claim for \removedS{$G'$}\change{$G$} by
	using $\substack{1\\1\,2}\stackrel{\small{P_2}}\longrightarrow\substack{1\\1\,2}$ for whimsical tiles, and alternate between two colorings using $P_s$.

 {\it Case 2:} Now, assume there is an even number of whimsical tiles and only a single sincere tile.
	We use propagation $P_2$ on all but $3$ consecutive whimsical tiles.
	These remaining whimsical tiles together propagate
  $\substack{1\\1\,2}\stackrel{P_w}\longrightarrow\substack{4\\1\,4}\stackrel{P_w}\longrightarrow\substack{3\\3\,4}\stackrel{P_w}\longrightarrow\substack{1\\1\,3}$.
  Together with $P_s$ for the sincere tile, we obtain the claimed propagation.

 {\it Case 3:} Finally, assume we have an even number of whimsical tiles and at least $3$ sincere tiles.
	Using $P_2$ for each whimsical tile, we only have to prove that we can construct a $\substack{1\\1\,2}\longrightarrow\substack{1\\1\,2}$-propagation for the sincere tiles.
	This is obtained by applying $P_s$ to all but $3$ of these tiles and using the following propagation on the remaining ones: $\substack{1\\1\,2}\stackrel{P_s}\longrightarrow\substack{1\\1\,3}\stackrel{P_s}\longrightarrow\substack{1\\1\,4}\stackrel{P_s}\longrightarrow\substack{1\\1\,2}$.

Thus, any sufficiently \ltwocc $G$ can be colored with $\Delta(G)$ colors and is first class.
\end{proof}

\section{Treewidth}
\label{sc:treewidth}

Treewidth is a central measure in graph theory and parameterized complexity~\cite{Bod07}.
It was first introduced by Bertelé and Brioschi under the term \emph{dimension} but rediscovered twice in following years~\cite{BB72, Hal76, RS84}.
Robertson and Seymour coined the term \emph{treewidth} and discovered a profound theory based on it that spawned a plethora of results.
While it is known that the treewidth of $2$-crossing-critical graphs is bounded from above by $2^{15\,361}-2$ \cite{Hli03}, the known bound is far from optimal.
Lower bounds are known for $k\ge 3$ only \cite{Hli03}.

\begin{definition}%
    A \emph{tree decomposition} of a connected graph~$G$ is a tree~$T$
    and a function $f \colon V(T) \rightarrow 2^{V(G)}$ such that
    \begin{enumerate}[(1)]
        \item for each edge~$uv \in E(G)$, there exists a vertex~$\alpha \in V(T)$ with $\{u,v\} \subseteq f(\alpha)$, and
        \item for each vertex~$v \in V(G)$, the subgraph of $T$ induced by $\big\{\alpha \in V(T) : v \in f(\alpha)\big\}$ is connected.
    \end{enumerate}
	Each set $f(\alpha)$ is typically called a \emph{bag}.
    The \emph{treewidth~$\tw(G)$} is the smallest~$\gamma \in \mathbb N$,
    such that there exists a tree decomposition of $G$ with $\max_{\alpha\in V(T)} |f(\alpha)| \leq \gamma + 1.$
\end{definition}

While the above definition is the classical one by Robertson and Seymour, there are several equivalent characterizations of treewidth.
For our proofs, we use one by Seymour and Thomas that employs a game of cops and robber~\cite{ST93}:
The cops and the robber stand on vertices of the graph $G$. The robber may move---at infinite speed---to any other vertex~$w$
    unless every path from $v$ to $w$ contains a vertex with a cop located on it.
Cops move by ``helicopter''\!, i.e., they are removed from their vertex and---at a later point in time---are placed on some other vertex.
All participants know all positions and the graph at all times.
The cops win if there exists a strategy such that after a finite number of cop movements,
    one of them is placed on the same vertex as the robber, independent of the robber's strategy.
Otherwise, the robber has a strategy to avoid being caught indefinitely and wins.
The treewidth of~$G$ is equal to the maximum number of cops such that the robber still wins.
The intuitive connection between the original treewidth characterization and this game-theoretic approach is that cops would block all vertices of a bag in the decomposition tree $T$, locking the robber in some subtree of $T$; then, the cops can iteratively move over to the adjacent bag that is closer to the robber, essentially pushing the robber towards a leaf-bag, where he will eventually be catched. If there are too few cops, they will be unable to always lock the robber within a subtree, and the robber can flee ad infinitum. See~\cite{ST93} for details.

Similarly, since treewidth is minor-monotone, one may characterize graphs of treewidth at most $k$, also called \emph{partial $k$-trees},
    by a set of forbidden minors~\cite{APC90, RS04}.
Since all \twocc{}s are non-planar, it follows from Kuratowski's theorem that they contain the $K_4$ as a minor
    and their treewidth is at least $3$~\cite{Bod88}.

\begin{definition}%
    A tile~$T=(G,x,y)$ is \emph{blocked}, if there are cops on $G$ such that
        the robber---independent of his position---
        cannot move from a vertex on $T$'s left wall $x$ to a vertex on $T$'s right wall $y$ while using only edges of $T$, i.e.,
        the graph~$G[W]$ induced by the vertices~$W$ of $G$ that are not occupied by a cop,
        contains no path from a vertex in $x$ to a vertex in $y$.
\end{definition}

\begin{lemma}
    Any \ltwocc~$G$ with at least $5$ elementary tiles has $4 \leq \tw(G) \leq 5$.
\end{lemma}
\begin{proof}
    Recall that the generalized Wagner graph~$V_8$ is a cubic graph that is constructed from the cycle on $8$ vertices $v_1,\ldots,v_8$ (in this order) by adding the edges $v_iv_{i+4}$, $1\leq i\leq 4$ \change{(cf. \Cref{fig:MoebiusWagner} for the analogously defined $V_{10}$)}.
    The $V_8$ constitutes one of four obstructions in the characterization of graphs with
        treewidth~$\leq 3$ \cite{APC90, ST90}.
    We obtain $\tw(G) \geq 4$ since any $G$ with at least $5$~elementary tiles contains the~$V_8$ as a .

    Let us now describe a strategy for catching a robber on any $G$ with $6$ cops:
        Using $2$ cops, we may block any tile~$T$ by placing them on its left wall.
    Applying this operation iteratively, using $3$ sets of $2$ such cops each, we can force the robber into a single elementary tile~$T'$
        (essentially using binary search), such that there is a cop on each wall vertex of~$T'$.
    Checking each possible elementary tile individually, one can see that catching the robber within $T'$ is then always possible with $6$ cops.
\end{proof}

In fact, also the graphs on $3$ elementary tiles contain $V_8$ as a
	minor unless each tile has the signature~$\sigma\tau\fL$ with $\sigma,\tau \in \{\pA,\pD\}$.
A treewidth-$3$ decomposition for the latter cases is easily obtained.\footnote{The reader may check the central case $\sig(G)=\pA\pA\fL\,\pA\pA\fL\,\pA\pA\fL$ either by hand or, e.g., using ToTo~\cite{Wer17}. The other cases follow since they are minors of this $G$.}
It remains to distinguish the large graphs with treewidth~$4$ from those with treewidth~$5$.
Surprisingly, for this we only need to recognize one specific minor~$\hg^3$:

\begin{definition}%
    The \emph{hourglass graph}~\hg is obtained from two disjoint triangles
    by identifying one vertex from the first with one vertex from the second triangle.
	The graph $\hg^3$ is obtained by \change{cyclically} joining three hourglass graphs, as \removedS{visualized}\change{given} in \Cref{fig:forbidden-minor}.
\end{definition}
\begin{figure}
	\centering
	\winIgnore{\includegraphics[scale=0.7]{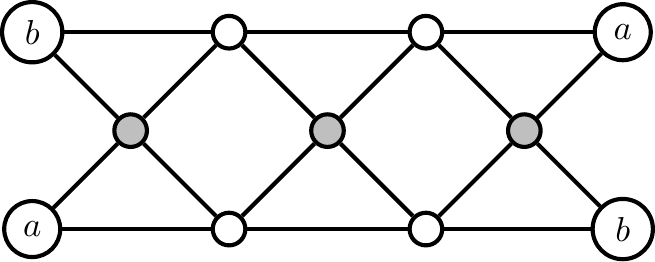}}
	\caption{
		Drawing of the minimal forbidden minor~$\protect\hg^3$ for treewidth~$4$ in the projective plane as indicated by the identification of the same-labeled vertices.
		Vertices with the same labels are identified.
		Internal vertices are gray.
	}
	\label{fig:forbidden-minor}
\end{figure}

We now consider a general refined cop strategy that will allow us to use less than $6$ cops in some cases.

\begin{definition}%
	Consider a tile~$T=(G,x,y)$ in a \ltwocc with cops placed on some of its vertices.
	Let $W \subseteq V(G)$ denote the set of vertices occupied by cops.
	A vertex $v \in V(G)$ is \emph{left-blocked} (\emph{right-blocked}) if $G[V(G) \setminus W]$ contains no
	path from $v$ to a vertex of its left wall (right wall, respectively).
	A vertex of~$G$ is \emph{tracked} if itself or all its neighbors are occupied by cops.
	A \emph{sweep} of~$T$ is a sequence of cop movements such that
		(1)~the cops initially occupy the left wall,
		(2)~after the sequence, the cops occupy the right wall; and
		(3)~during the cop movements, each vertex in $V(G)$
		    enters the three states ``left-blocked'', ``tracked'', and ``right-blocked'' in that order such that
		    each state is entered exactly once and at each point in time, at least one state applies.
\end{definition}

Observe that during a sweep, the respective tile always remains blocked
since there is no vertex that is connected to both its walls.
Further, a sweep is in fact applicable in both directions, i.e., the reverse sequence allows cops to move from the right to the left wall in the same manner.

\begin{definition}%
    An elementary tile is \emph{messy}, if it contains~\hg as a minor.
    An elementary tile~$T$ is \emph{neat} if there is a sweep of~$T$ that uses at most $3$ cops.
\end{definition}

\begin{lemma}
	\removedS{An}\change{Each} elementary tile is \removedS{either}neat or messy.
\end{lemma}
\newcommand{\dotprec}{\ \dot\prec\ }
\begin{proof}
    We show that elementary tiles with pictures~$\mathcal M \coloneqq \{\pV\pA,\pV\pI\pA,\pB\pA,\pB\pI\pA,$ $\pH\}$ are messy and the remaining ones are neat.
    For this proof, we denote by $X \dotprec Y$ that picture $X$ is a minor of picture $Y$ such that %
    $X$ and $Y$ have the same frame vertices.

    For the first part, we contract all but the center $4$-cycle of each tile's frame.
    Picture~$\pH$ becomes~\hg by contraction of the central edge.
    Clearly, $\pV\pI\pA \dotprec \pV\pA$ and similarly $\pB\pI\pA \dotprec \pB\pA$---each by contraction of a single edge.
    Also, $\pV\pI\pA \dotprec \pB\pI\pA$ by contracting the double edge in $\pB$.
    Hence, all elements of $\mathcal M$ that are not $\pH$ contain $\pV\pI\pA$ as a minor that by removal of a single edge becomes~\hg.

    It remains to show that tiles with pictures not in $\mathcal M$ are neat.
    These pictures are $\overline{\mathcal M} \coloneqq \{\pD\pD,\pD\pV,\pD\pB,\pD\pA,\pV\pV,\pV\pB,\pB\pB,\pA\pA\}$.
    For the sweep, we may assume to start on the picture's vertices as it is trivial to move from any wall-vertex that is not part of the picture to its adjacent vertex in the picture.
    Note that we may also omit the mirrored pictures $\pD\pA$ and $\pA\pA$ since it suffices to show a sweep of their
        mirrored counterparts ($\pD\pV$ and $\pV\pV$, respectively) that are also in $\overline{\mathcal M}$.
    The remaining pictures in $\overline{\mathcal M}$ are all minors of $\pB\pB$:
        $\pD\pD \dotprec \pD\pV \dotprec \pD\pB$ and $\pV\pV \dotprec \pD\pB \dotprec \pB\pB$ where each minor-relation,
        except for the last, is witnessed by contraction of a single edge.
    Hence, it suffices to provide a sweep on $\pB\pB$:
        we label this picture's vertices as~$v_i$ starting at the top left with~$v_0$ in counter-clockwise order, see \cref{fig:bb}.
    Assuming the cops arrive from the left side, they occupy $v_0$ and $v_2$.
    First, we move the $3$rd cop to $v_3$.
    Observe that all neighbors of $v_1$ are now occupied by cops, i.e., $v_1$ is tracked.
    The remaining sweep goes as follows: $v_2 \rightarrow v_7, v_0 \rightarrow v_4, v_3 \rightarrow v_6$.
    Once again, $v_5$ is tracked by occupying its neighbors.
\begin{figure}
    \centering
    \winIgnore{\includegraphics[width=0.3\textwidth]{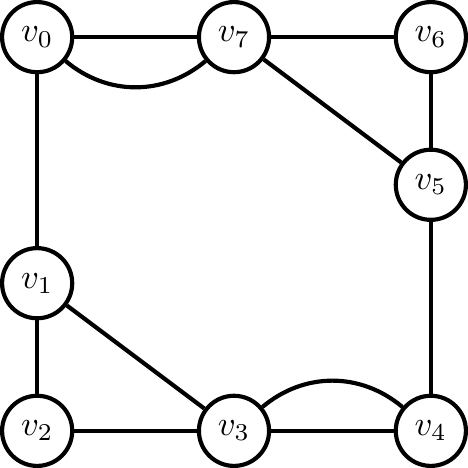}}
    \caption{
		Picture $\pB\pB$ with nodes labeled in counter-clockwise order.
	}
    \label{fig:bb}
\end{figure}
\end{proof}

\begin{theorem}\label{thm:treewidth}
    A large $2$-crossing-critical graph has treewidth~$5$ if and only if it contains at least $3$ messy tiles, i.e,
		if and only if it contains~$\hg^3$ as a minor.
\end{theorem}
\begin{proof}
    If there are no messy tiles, every elementary tile is neat and can be sweeped by $3$ cops.
    Hence, we may block an arbitrary tile with $2$ cops and sweep around the remaining graph with $3$ cops.
    If during the sweep, the robber should remain on a vertex that is tracked by cops occupying its neighbors,
    we catch him using one of the first $2$ cops.

    Similarly, if there are up to two messy tiles, say $X$ and $Y$, $4$ cops are initially placed on $X$'s walls.
    Then, the cops on the left wall of $X$ start to sweep using a $5$th cop~$c$ until they reach the wall of $Y$.
    Finally, the cops on the right wall of $X$ sweep, again using~$c$, until they reach the other wall of $Y$.

    If, on the other hand, there are $3$ messy tiles in $G$, then $G$
        contains the forbidden minor~$\hg^3$, as witnessed by contracting all edges that
        do not belong to the set of $3$ messy tiles and contracting each messy tile to~\hg.
    On $\hg^3$, however, there exists a simple strategy for the robber to win against $5$ cops:
        There are two types of vertices in~$\hg^3$: $6$ \emph{rim} vertices and $3$ \emph{internal} ones,
        seen in \cref{fig:forbidden-minor} as the (white) top/bottom and (gray) middle ones, respectively.
    The robber stays on an arbitrary rim vertex~$u$ until the last of its neighbors, say $v$, is about to be occupied by a cop.
    It then moves over $v$ to a new rim vertex that is not adjacent to $u$.

    If $v$ is an internal vertex, it is adjacent to $4$ rim vertices:
        $u$, a neigbor of $u$ and two other vertices $w_1,w_2$ that are not adjacent to $u$.
    Since there are only $5$ cops, $w_1$ or $w_2$ is not occupied and the robber may move to it.
    Conversely, if $v$ is a rim vertex, the robber will, depending on the position of the remaining fifth cop,
        either move another edge along the rim or over the non-occupied internal vertex to a further rim vertex non-adjacent to $u$.
    Since any pair of non-adjacent rim vertices has exactly two common neighbors,
        not all neighbors of the new rim vertex are occupied even after the cop lands on $v$.
\end{proof}
\begin{corollary}
    \label{cor:bounded_treewidth}
    Any \ltwocc{} on at least $5$ elementary tiles has treewidth $5$ if and only if it contains at least three elementary tiles with pictures from the set $\{\pV\pA,\pV\pI\pA,$ $\pB\pA,$ $\pB\pI\pA, \pH\}$. Otherwise, it has treewidth $4$.
\end{corollary}

\section{Conclusions}

For several graph classes, we have conjectures on their crossing numbers. But there are only very
few classes for which we know their crossing numbers. Then, their structure is mostly rather simplistic.
The class of 2-crossing-critical graphs seems to be the first graph class with known crossing
numbers that still offers rich and non-trivial structure in terms of other graph measures as well.

In this paper, after some straight-forward graph properties as building blocks, we successfully discussed
both their chromatic number and index, as well as their treewidth. We propose further investigation
of general graph-theoretic properties of \removedS{infinite}crossing-number related \change{infinite} graph families, to further
the idea of interlinking the concepts of topological graph theory with other aspects of the field
and further discovery of new applications.
---
On a more specific note, we recall Question \ref{q:character3} from above, which asks whether we can fully characterize 3-colorable \ltwocc{}s.

	In all our proofs, knowing the structure of large $2$-crossing-critical graphs was instrumental to proving the values of the above invariants. For further research, it would be of interest
	to obtain these values without referring to the structure of the
	graphs, possibly by just assuming the $3$-connectivity,
	$2$-crossing-criticality and (should it be needed), presence of
	a $V_{10}$ subdivision. Such approaches to graph invariants on $2$-crossing-critical graphs may then be generalizable to $c$-crossing-critical graphs for $c>2$.
	Furthermore, there are other graph invariants and problems one could consider on these graphs, \texttt{www.graphclasses.org} sharing an extensive list. By investigating these invariants
	and specifically by obtaining proofs that require no knowledge about the structure of the underlying $2$-crossing-critical graphs, one may
	find ways to simplify the characterization theorem of \cite{Bok+16}, or to identify an approach that would allow to
	list the finitely many $2$-crossing-critical graphs that contain a $V_{8}$, but not a $V_{10}$ subdivision,
	which is the final open
	step that would render their characterization completely constructive.

\section*{Acknowledgement}
D.B.\ was funded in part by Slovenian Research Agency ARRS, grant J1--8130 and programme P1--0297.
In addition, the research was initiated during knowledge exchange visit within the project INOVUP
funded by the Republic of Slovenia and the European Union from the European Social Fund. M.C.\ and T.W.\ were partially funded by the German Research Foundation DFG, project CH 897/2-2.

\clearpage
\bibliographystyle{abbrvurl}
\bibliography{main}

\clearpage
\appendix
\section*{Appendix}

\begin{figure}[H]
	\centering
	\winIgnore{\includegraphics[width=\textwidth]{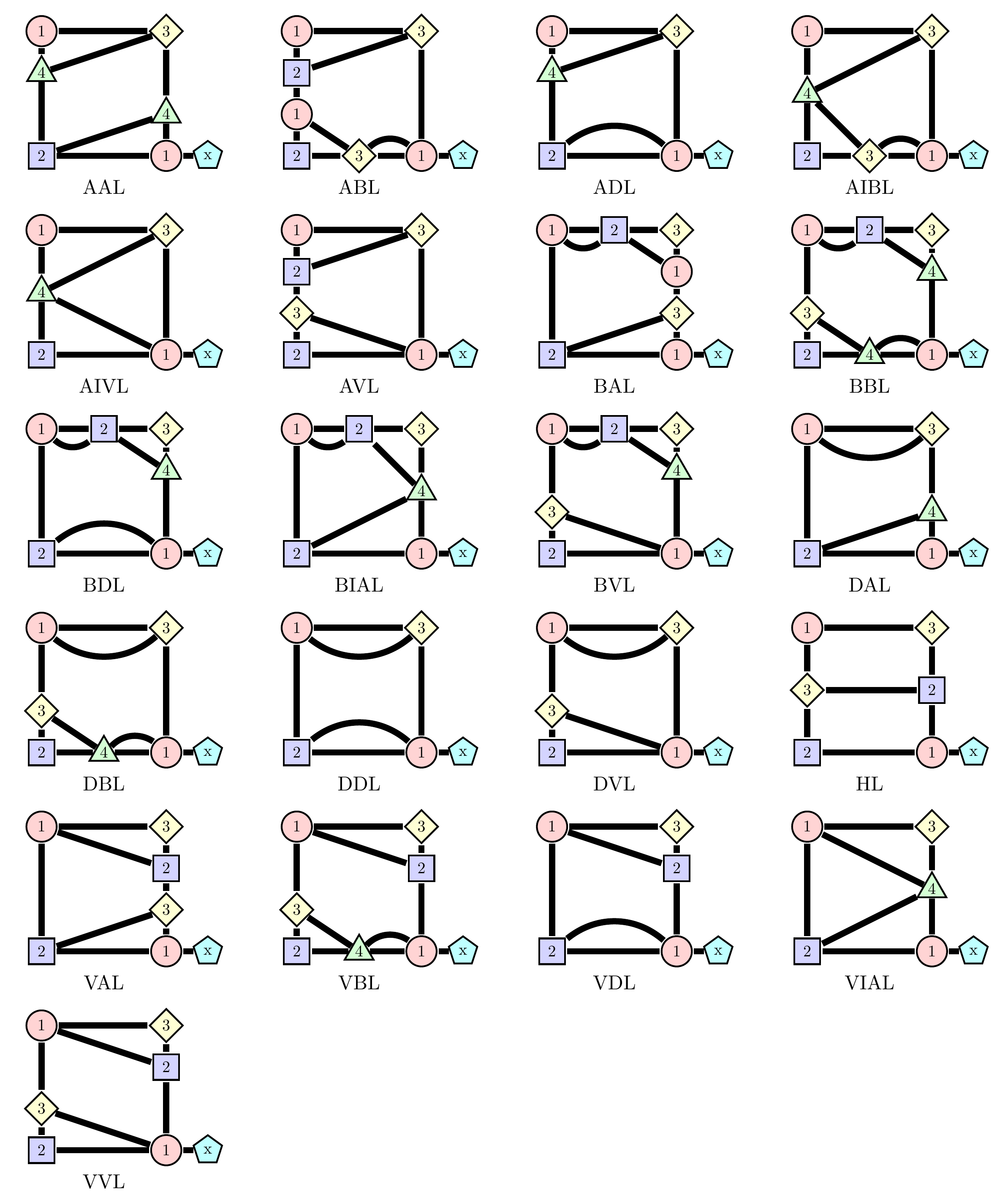}}
	\caption{4-vertex-coloring of every tile of $\cS$~with \fL-frame with
		$\protect\substack{a\\b}\rightsquigarrow\protect\substack{c\\!a}$-propagation. The vertex marked as $x$ can be colored in any color other than~$1$.}
	\label{Figures4Colorable2}
\end{figure}

\begin{figure}
	\centering
	\winIgnore{\includegraphics[width=\textwidth]{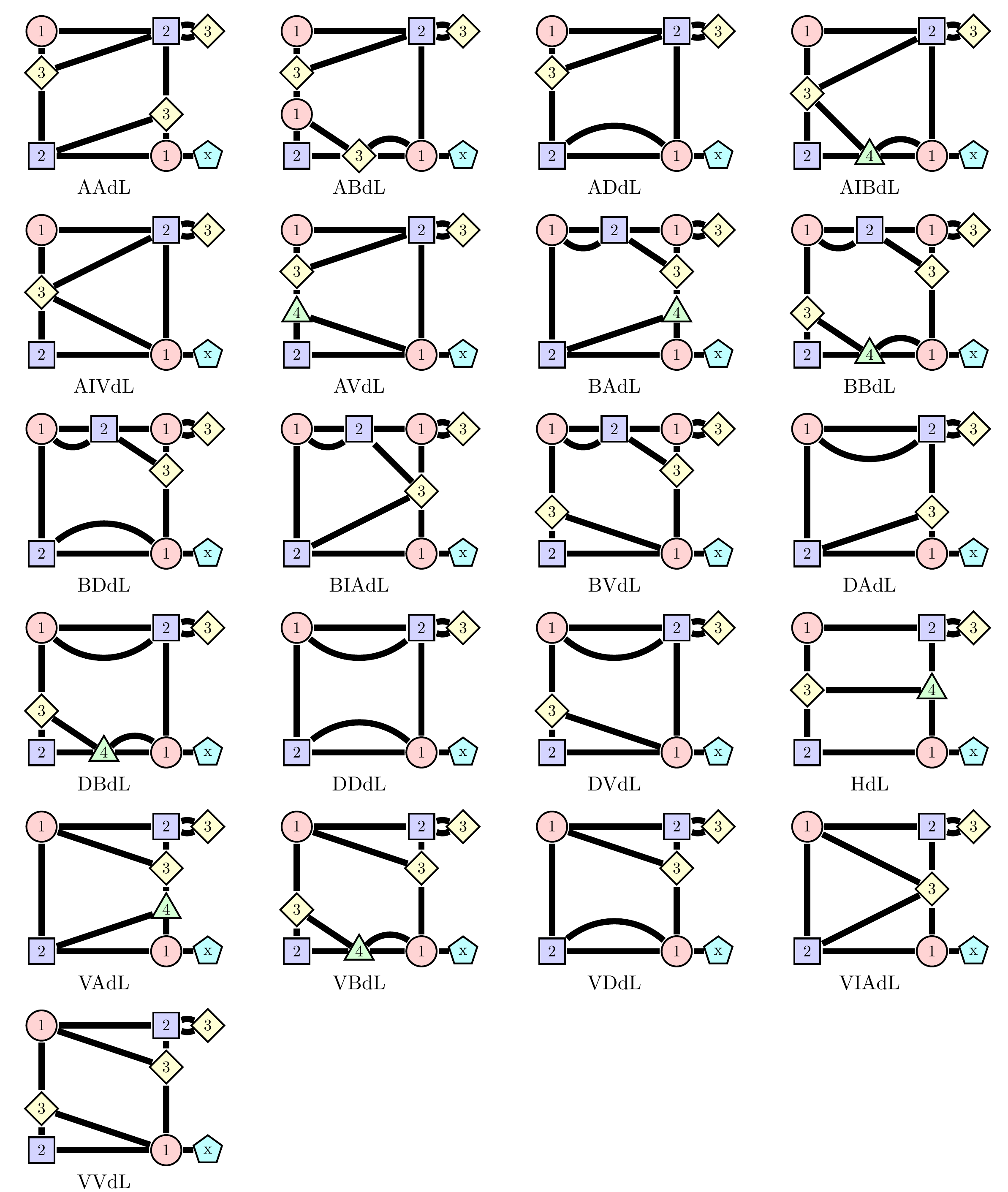}}
	\caption{4-vertex-coloring of every tile of $\cS$~with \fdL-frame with
		$\protect\substack{a\\b} \rightsquigarrow \protect\substack{c\\!a}$-propagation. The vertex marked as $x$ can be colored in any color other than~$1$.}
	\label{Figures4Colorable1}
\end{figure}

\begin{figure}[H]
 \centering
 \winIgnore{\includegraphics[width=\textwidth]{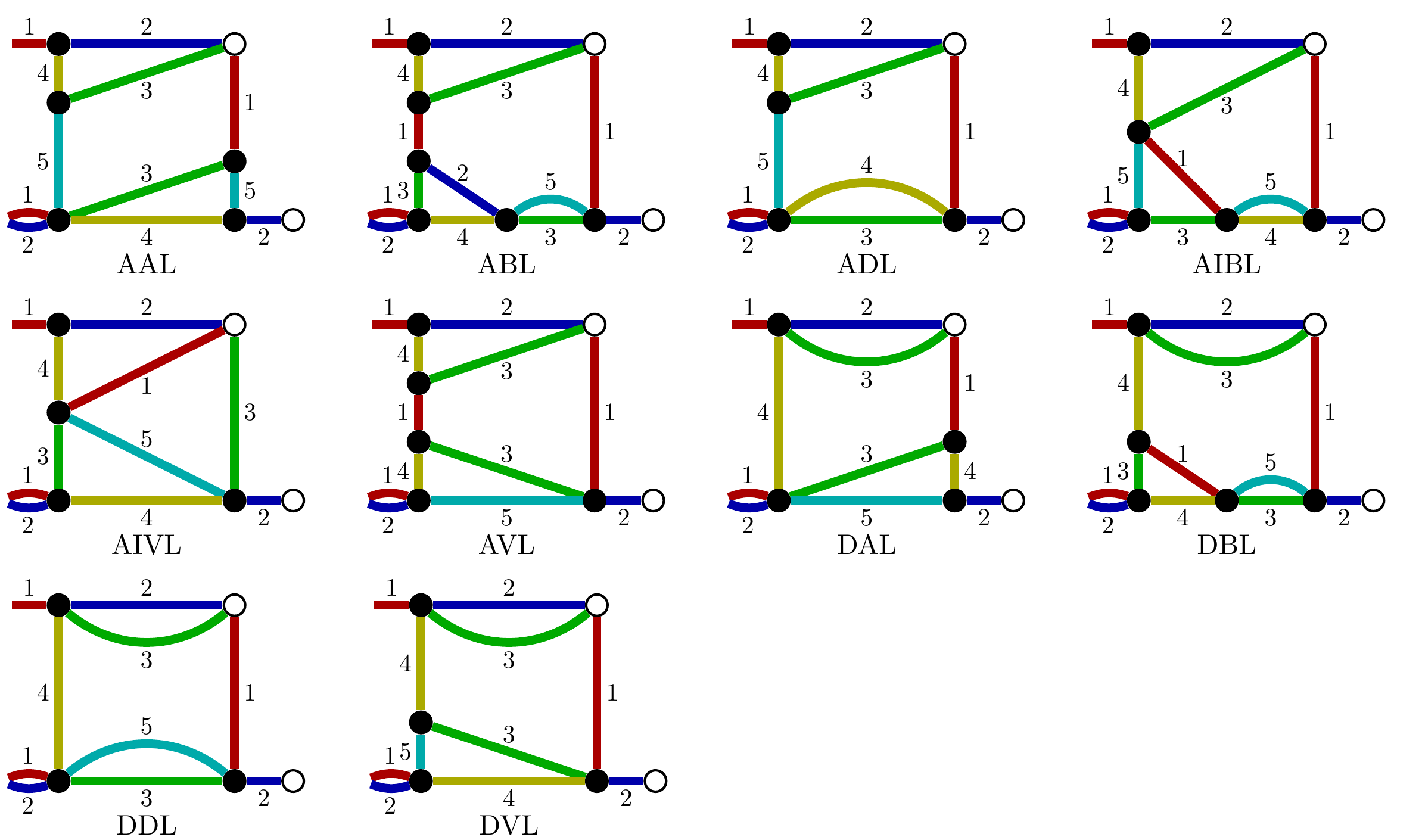}}
 \caption{Edge color propagation $P_{23}$ for all $2$-$3$-tiles.}
 \label{fig:ec_p23}
\end{figure}

\begin{figure}[h]
 \centering
 \winIgnore{\includegraphics[width=\textwidth]{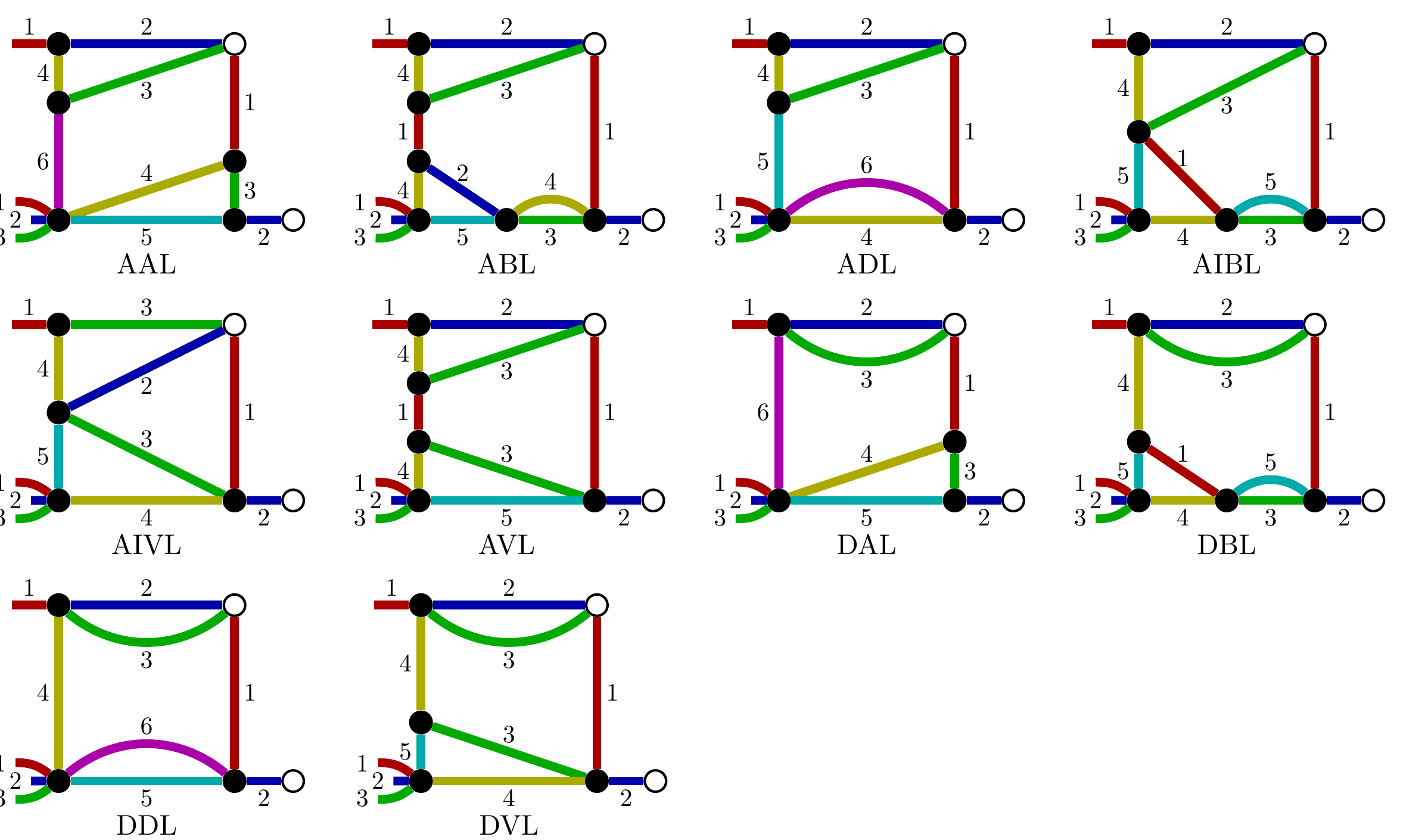}}
 \caption{Edge color propagation $P_3$ for all $3$-edge-tiles.}
 \label{fig:ec_p3}
\end{figure}

\begin{figure}
 \centering
 \winIgnore{\includegraphics[width=\textwidth]{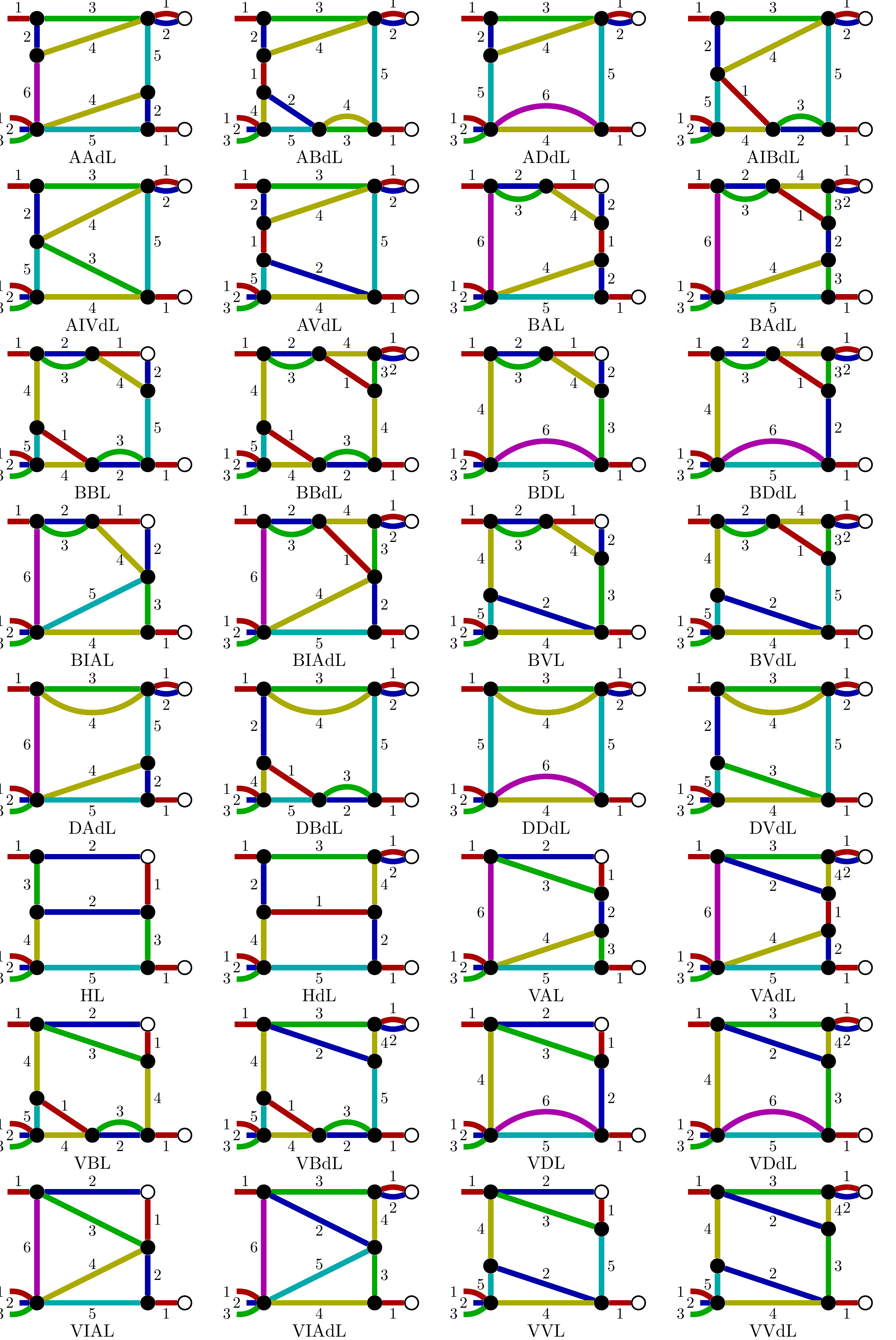}}
 \caption{Edge color propagation $P_{32a}$ for all $3$-$2$-tiles.}
 \label{fig:ec_p32a}
\end{figure}

\begin{figure}
 \centering
 \winIgnore{\includegraphics[width=\textwidth]{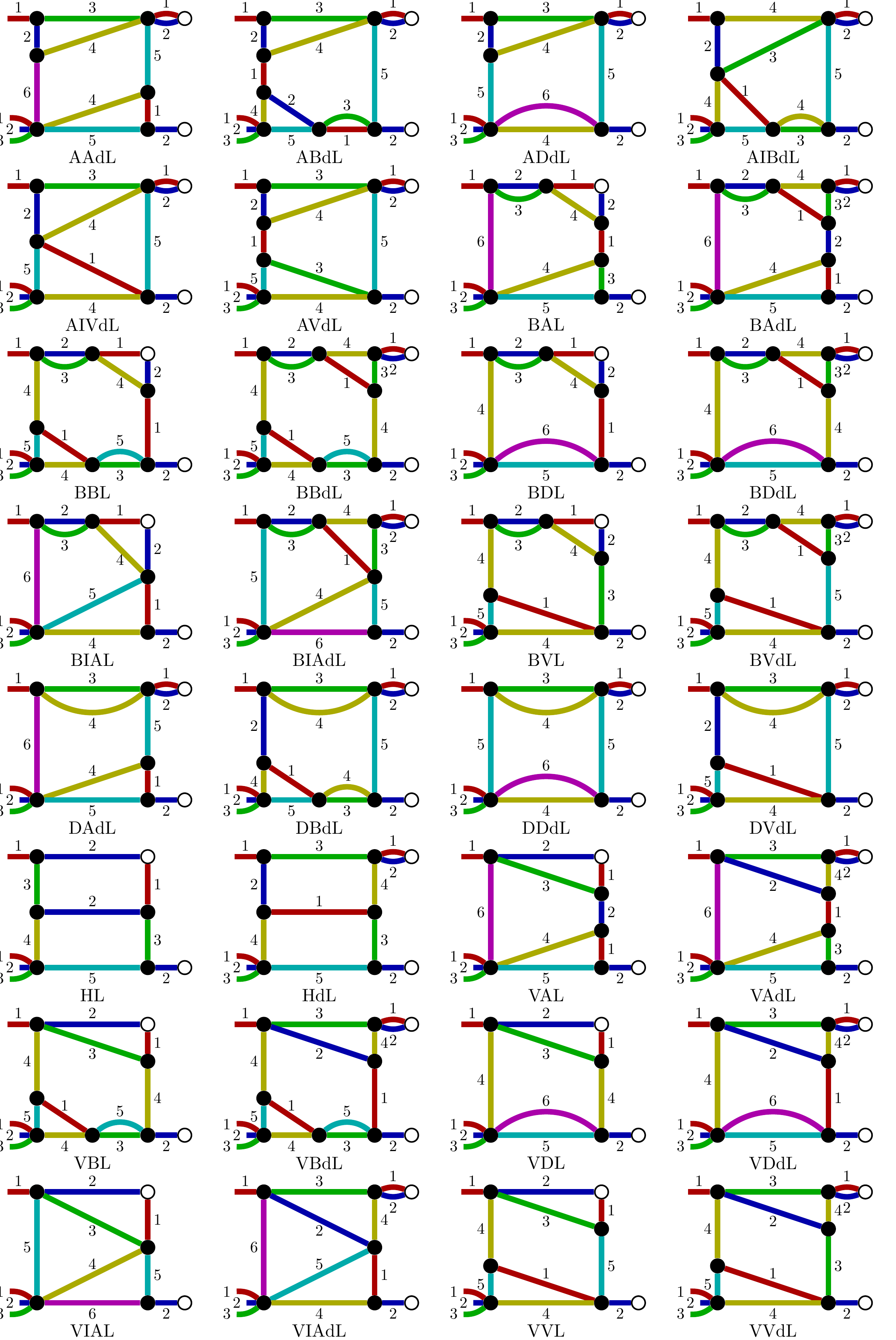}}
 \caption{Edge color propagation $P_{32b}$ for all $3$-$2$-tiles.}
 \label{fig:ec_p32b}
\end{figure}

\begin{figure}
 \centering
 \winIgnore{\includegraphics[width=\textwidth]{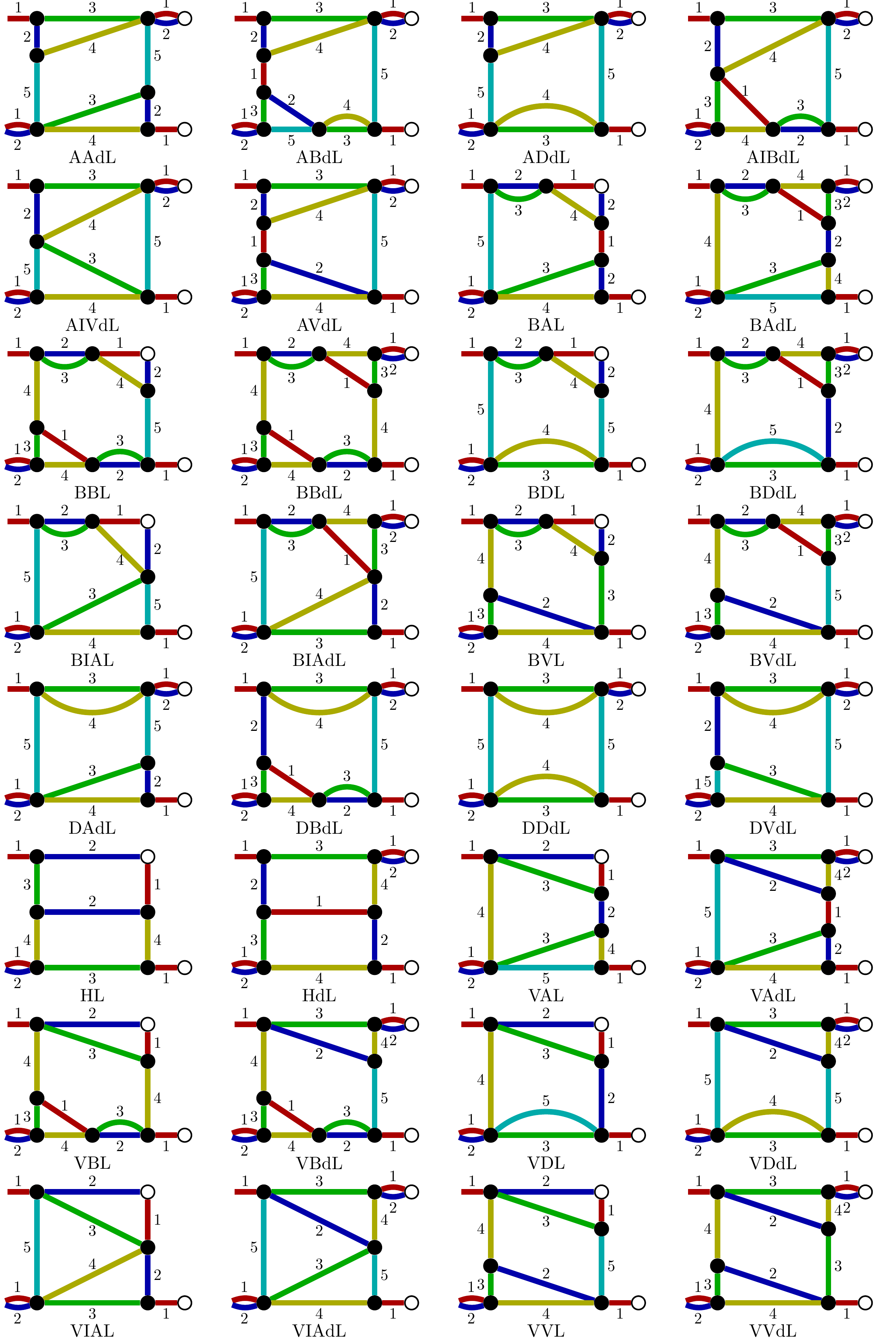}}
 \caption{Edge color propagation $P_2$ for all $2$-edge-tiles. Note that $\pB\pB\fL$, $\pV\pV\fL$, $\pB\pV\fdL$, $\pV\pB\fdL$ require $5$ colors for $P_2$ despite having maximum degree $4$.}
 \label{fig:ec_p2}
\end{figure}

\begin{figure}
 \centering
 \winIgnore{\includegraphics[width=\textwidth]{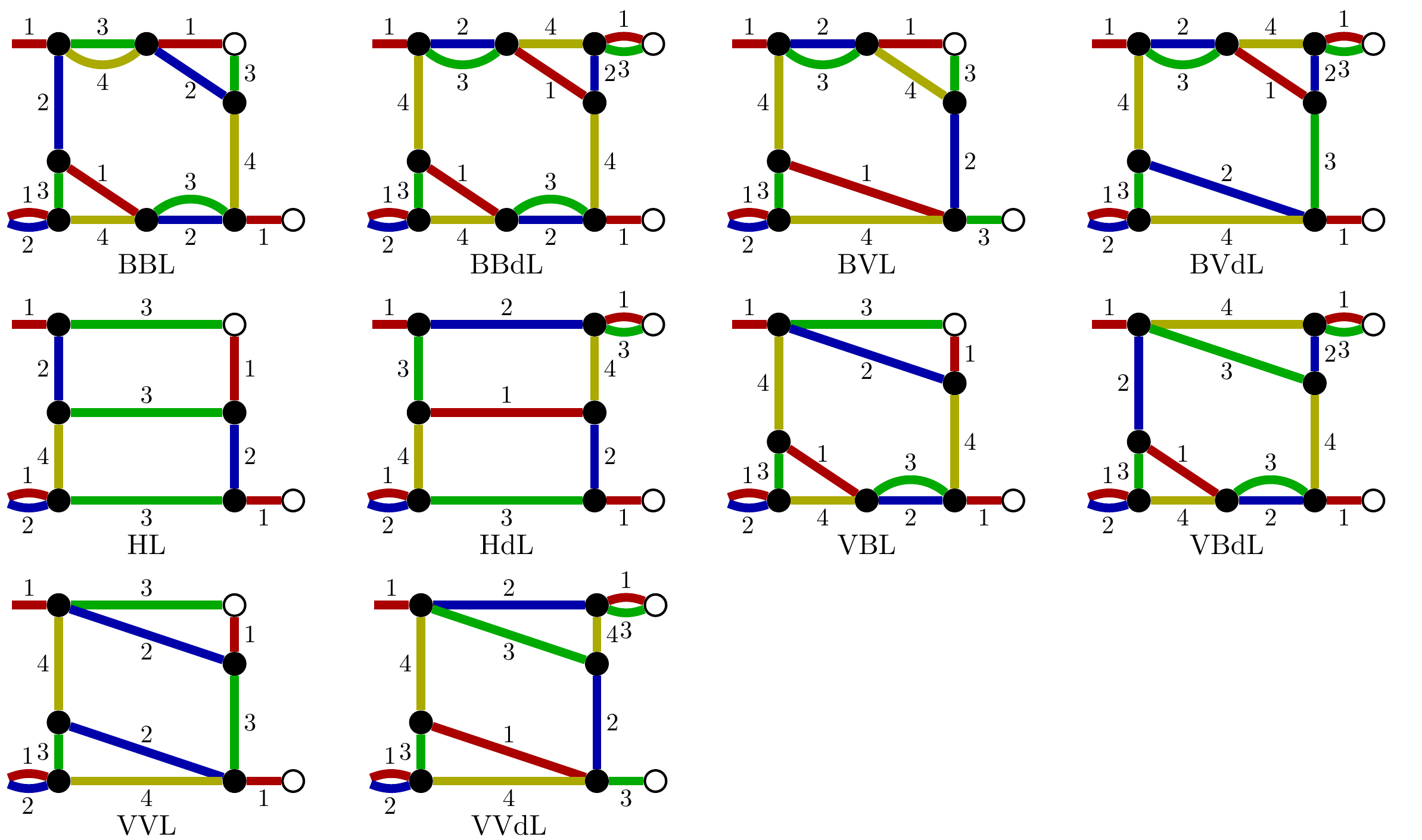}}
 \caption{Edge color propagations $P_w$ for $\pB\pV\fL$ and $\pV\pV\fdL$ and $P_s$ for all others of maximum degree $4$.}
 \label{fig:ec_fourcolors}
\end{figure}

\end{document}